\documentclass[lettersize,conference,10pt]{IEEEtran}
\IEEEoverridecommandlockouts
\usepackage[compact]{titlesec}
\titlespacing{\section}{0pt}{2ex}{1ex}
\titlespacing{\subsection}{0pt}{1ex}{0ex}
\usepackage{cite}
\usepackage{flushend}
\usepackage{empheq} 
\usepackage{amsmath,amssymb,amsfonts,cases}
\usepackage{dsfont}
\usepackage{hyperref}
\usepackage{amsthm}
\usepackage{algorithmic}
\usepackage{graphicx}
\usepackage{textcomp}
\usepackage[caption=false]{subfig}
\usepackage{setspace}
\usepackage{float}
\usepackage{xcolor}
\usepackage{mathtools}
\usepackage{cite}
\usepackage{times}
\usepackage{color}
\usepackage{stackrel, url}
\usepackage{lipsum, enumitem}
\theoremstyle{remark}
\newcommand\scalemath[2]{\scalebox{#1}{\mbox{\ensuremath{\displaystyle #2}}}}

\newtheorem{theorem}{\bf \emph{Theorem}}
\newtheorem{lemma}{\bf \emph{Lemma}}

\newtheorem{proposition}{\bf \emph{Proposition}}
\newtheorem{remark}{\textit{Remark}}

 \def\cF{{\mathcal{F}}}  
  \def\cK{{\mathcal{K}}} \def\cL{{\mathcal{L}}}
\def\cM{{\mathcal{M}}} \def\cN{{\mathcal{N}}} \def\cO{{\mathcal{O}}}

\def\ba{{\mathbf{a}}} \def\bb{{\mathbf{b}}}  \def\bd{{\mathbf{d}}} \def\be{{\mathbf{e}}} 
\def\bff{{\mathbf{f}}}    
   \def\bn{{\mathbf{n}}} 
   \def\bs{{\mathbf{s}}} \def\bt{{\mathbf{t}}}
\def\bu{{\mathbf{u}}} \def\bv{{\mathbf{v}}}  \def\bx{{\mathbf{x}}} \def\by{{\mathbf{y}}}

\def\bA{{\mathbf{A}}} \def\bB{{\mathbf{B}}} \def\bC{{\mathbf{C}}} \def\bD{{\mathbf{D}}} 
\def\bF{{\mathbf{F}}} \def\bG{{\mathbf{G}}} \def\bH{{\mathbf{H}}} \def\bI{{\mathbf{I}}}

\def\bU{{\mathbf{U}}} \def\bV{{\mathbf{V}}} \def\bW{{\mathbf{W}}} \def\bX{{\mathbf{X}}}

\def\N{{\mathbb{N}}}  \def\R{{\mathbb{R}}} \def\C{{\mathbb{C}}}       \def\E{\mathbb{E}}

\def\blockdiag{\mathop{\mathrm{blk}}}

\def\argmin{\mathop{\mathrm{argmin}}}

\def\det{\mathop{\mathrm{det}}}

\def\tr{\mathop{\mathrm{tr}}}

\def\diag{\mathop{\mathrm{diag}}}

\def\mod{\mathop{\mathrm{mod}}}

     \def\d4{\!\!\!\!}

\def\bSig{\mathbf{\Sigma}}   \def\bLam{\mathbf{\Lambda}}

  \def\-{\! - \!}  \def\+{\! + \!}  \def\={\! = \!}  \def\>{\! > \!}

\newcommand{\bef}{\begin{figure}}
\newcommand{\eef}{\end{figure}}
\newcommand{\beq}{\begin{eqnarray}}
\newcommand{\eeq}{\end{eqnarray}}

\def\sumK{\sum_{k=1}^{K}}

\def\sumN{\sum_{n=1}^{N}}
\def\sumM{\sum_{m=1}^{M}}

\begin{document}
\title{Joint Delay-Phase Precoding Under True-Time Delay Constraints in Wideband Sub-THz Hybrid Massive MIMO Systems}
\author{Dang~Qua~Nguyen,~\IEEEmembership{Student Member,~IEEE} and Taejoon~Kim,~\IEEEmembership{Senior Member,~IEEE} \\
\thanks{Parts of this work was presented in the {IEEE} International Conference on Communications (ICC) 2022 \cite{Qua2022}. The preprint of this work can be found at \cite{Qua2023}. 
Dang Qua Nguyen is with the Department of Electrical Engineering and Computer Science, The University of Kansas, Lawrence, KS, 66045 USA (e-mail: quand@ku.edu).
Taejoon Kim is with the School of Electrical, Computer and Energy Engineering, Arizona State University, Tempe, AZ  85287 USA (e-mail: taejoonkim@asu.edu).

This work was supported in part by the National Science Foundation (NSF) under Grant CNS2212565, Grant CNS2225577, Grant ITE2226447, and Grant CNS1955561, and the Office of Naval Research (ONR) under Grant N000142112472, and the NSF and Office of the Under Secretary of Defense (OUSD) – Research and Engineering, Grant ITE2326898, as part of the NSF Convergence Accelerator Track G: Securely Operating Through 5G Infrastructure Program.
}
}
\maketitle	
\begin{abstract}
    In wideband sub-Terahertz (sub-THz) massive multiple-input multiple-output (MIMO) communication systems, the beam squint effect manifests as a substantial degradation in array gain.
    To mitigate the aforementioned beam squint effect, a hybrid precoding approach leveraging both true-time delay (TTD) and phase shifters (PS) has been proposed.
    However, existing methods operate under the assumption that the TTD device can generate any desired time delay value. 
    These methods subsequently design the TTD precoder while fixing the PS precoder.
    This work presents a novel optimization framework for the joint TTD and PS precoder design, incorporating realistic time delay constraints for each TTD device.
    Unlike previous methods, our framework does not rely on the unbounded time delay assumption and optimizes the TTD and PS values jointly to cope with the practical limitations.
    Furthermore, within the context of our proposed framework, we mathematically determine the minimum number of TTD devices necessitated to achieve a predetermined target array gain.
    Simulations confirm the proposed approach exhibits performance improvement, guarantees array gain, and achieves computational efficiency.
    \end{abstract}
\section{Introduction}
\label{secI}
       Sub-Terahertz (sub-THz) band (90-300 GHz) communication is a potential technology for the sixth generation (6G)-\&-beyond wireless systems, which are expected to support various high data rate applications such as augmented reality/virtual reality (AR/VR), eHealth, and holographic telepresence \cite{Giordani2020}. 
        The sub-THz band offers abundant spectrum resources with tens of gigahertz (GHz) bandwidths \cite{Han2014, Rappaport2019, ChenHan2021, Akyildiz2022}, compared to the current millimeter-wave (mmWave) bands in the 5G specifications \cite{3gpp.21.917} that utilize a few GHz bandwidths. 
        This enables the possibility of achieving the data rates on the orders of 10 to 100 Gbps using the existing digital modulation techniques in sub-THz frequencies \cite{Rappaport2019, Koenig2013}. 
        However, the sub-THz communication faces several challenges such as high path losses, large power consumption, and inter-symbol interference.
        To overcome these challenges, hybrid massive multiple-input multiple-output (MIMO) and orthogonal frequency-division multiplexing (OFDM) technologies have been widely investigated recently \cite{dai2021,Han2022,Gao2021,Matthaiou2021,Yezeng2023,Yuan2023,Elbir2021,Elbir2022}.
        Massive MIMO technology has received significant attention previously \cite{Marzetta2010}.  
        Strong theoretical analyses have justified the use of a very large number of antennas at the base station \cite{Marzetta2010, Larsson14, Rusek2013}. 
        This has raised significant interest in massive MIMO systems at sub-$6$ GHz \cite{Marzetta2010,Rusek2013, Larsson14} and mmWave frequencies \cite{ayach2014,Alkhateeb2014,Hadi2016}.
        The underlying assumption behind \cite{Marzetta2010,Rusek2013,Larsson14,ayach2014,Alkhateeb2014,Hadi2016} was, however, narrowband. 
        Wideband high-frequency massive MIMO OFDM systems may experience significant array gain loss across different OFDM subcarriers due to the spatial wideband effect \cite{Cai2016, Liu2019, Wang2018, Wang20182, Wang2019, Han2021, tan2020}, also known as the beam squint effect. 
       Beam squint refers to the phenomenon in which the deviation occurs in the spatial direction of each OFDM subcarrier when the wideband OFDM is used in a very large antenna array system.
       This implies that beam squint can cause severe degradation of the achievable rate, which potentially demotivates the use of OFDM in the wideband sub-THz massive MIMO systems. 
       Therefore, an efficient beam squint compensation is essential for the realization of wideband sub-THz massive MIMO communications.
%\vspace{-0.5cm}       
\subsection{Related Works}
        Previous works on wideband mmWave massive MIMO \cite{Cai2016, Liu2019} addressed the beam squint by designing the beamforming weights to generate adaptive-beamwidth beams that cover the squinted angles. 
        However, these techniques are not applicable to the sub-THz bands due to the extremely narrow pencil beam requirement imposed by the much higher carrier frequencies \cite{Monroe2022,Yezeng2022survey, Han2021}. 
        Furthermore, several hybrid precoding approaches \cite{Elbir2021, Elbir2022} have attempted to address beam squint through digital signal processing techniques. These approaches design the digital precoder by projecting it onto the subspace spanned by the analog beamforming vectors. 
        However, this approach can only partially compensate for the beam squint due to the low-rank characteristic of the analog beamforming.\looseness=-1 
        The beam squint effect of phased array antennas has been independently studied in the radar community (e.g., see \cite{Mailloux2017, Longbrake2012, Rotman2016}, and references therein).
        A common method is to use true-time delay (TTD) lines instead of using phase shifters (PSs) for analog beamforming \cite{Mailloux2017, Longbrake2012,Rotman2016,Ghaderi2019}.
        Unlike the PS-based analog beamforming that produces frequency-independent phase rotations, the TTD-based analog beamforming generates frequency-dependent phase rotations that can be used for compensating the squinted beams in the spatial domain.
        However, this method is not directly applicable to massive MIMO systems because it requires a large number of TTD lines.
        Specifically, each transmit antenna needs to be connected to a dedicated TTD, resulting in a high hardware cost and huge power consumption\footnote{While the power consumption of a TTD device depends on a specific process technology (e.g., BiCMOS \cite{Cho2018}, and CMOS \cite{Hu2015}), a typical TTD in sub-THz consumes $100$ mW \cite{dai2021}. It is worth noting that the power consumption of a typical PS in sub-THz is $20$ mW \cite{dai2021}, which is much lower than that of a TTD.}. 
        TTD lines have been recently proposed for sub-THz hybrid massive MIMO OFDM systems \cite{dai2021, Gao2021, Matthaiou2021, Han2022, Yezeng2023}.
        These systems \cite{dai2021, Gao2021, Matthaiou2021} use a significantly smaller number of TTD lines than that in radar systems \cite{Mailloux2017, Longbrake2012, Rotman2016}. 
        These methods can be seen as combining a small number of TTD lines with a layer of PSs to form an analog precoder that is able to generate frequency-dependent phase rotations while consuming less power than the conventional TTD architecture in radar systems.
        However, these approaches still incur a large amount of power consumption when a large number of TTD lines is needed to combat the severe beam squint. 
        Determining the minimum number of TTDs to achieve a desired beam squint compensation capability remains challenging.\looseness=-1
        The TTD-based hybrid precoding architecture has also been utilized to address fast beam training \cite{Cui2022}, beam tracking \cite{Tan20212}, {frequency multiplexing \cite{Cabric2022, Ratnam2022, Jain2023}},  and user localization \cite{Gao2022} problems. 
        Unlike the beam squint compensation that aligns the beam directions at every OFDM subcarrier to a same spatial direction, these works \cite{Cui2022, Tan20212, Gao2022, Cabric2022, Ratnam2022, Jain2023} exploit the beam squint effect to spread the beams across different OFDM subcarriers simultaneously. 
        By tracking users' locations simultaneously, the communication overhead for channel sounding is substantially reduced \cite{Gao2022, Cabric2022, Ratnam2022, Jain2023}.\looseness=-1
Most of these previous works have focused on designing the TTD precoder while keeping the PS precoder fixed. 
This simplifies the analog precoding design by decoupling the PS and TTD precoders, and reduces the number of design variables since the number of deployed TTDs is usually less than the number of PSs.
    Additionally, these studies have assumed that the TTD values increase linearly with the number of antennas without any limits.
        This assumption is, however, unrealistic, as the range of the time delay values that a TTD can produce is strictly limited \cite{Qua2022}, \cite{Lin20221}, \cite{Lin20222}. 
        For instance, \cite{Hu2015} and \cite{Cho2018} design TTDs with the maximum time delay values of 400 ps and 508 ps, respectively. The former occupies a smaller circuit board area of 4 $ \text{mm}^2$ while the latter occupies 5.45 $ \text{mm}^2$. This raises a critical issue; for a given circuit board size (e.g., 128 $ \text{mm}^2$ \cite{Studer2020}), we need to limit both the time delay range and the number of TTDs used to mitigate the beam squint. 
        We note that hardware components are typically static once deployed. 
        Depending on the deployment scenarios, the time delay range of the deployed TTDs may be insufficient for beam squint compensation.
        It is sustainable to optimize the tunable parameters, such as TTD and PS values, to effectively deal with the beam squint instead of replacing them with new TTDs, which is costly.\looseness=-1

\subsection{Contributions and Synopsis}
    We propose a signal processing framework to mitigate the beam squint effect based on joint TTD and PS optimization.  
    The practical TTD constraint is imposed such that the time delay values are restricted in a fixed interval. 
    Our main contributions are summarized as follows.\looseness=-1

\begin{itemize}[leftmargin = 3mm]
%
            % \item 
            % We characterize the impact of beam squint compensation on the achievable rate performance. 
            % %%
            % \textcolor{blue}{We show that the ideal analog precoder that fully compensates for the beam squint is equivalent to the one that maximizes the achievable rate with an appropriate digital precoder. This motivates us to design the PS and TTD precoders to best approximate the ideal analog precoder.}
            %%
            \item
            {We first determine the ideal analog precoder that fully compensates for the beam squint.}
            Then, we formulate the joint TTD and PS precoder optimization problem based on minimizing the distance between the ideal analog precoder and the product of TTD and PS precoders under the TTD constraints. 
            Unlike the previous approaches \cite{dai2021, Gao2021, Matthaiou2021} that only optimize the TTD precoders under the unbounded time delay values assumption, we jointly optimize the TTD and PS precoders subject to the limited range of time delay values that a practical TTD can produce.
            Although the formulated problem is non-convex and 
            difficult to solve directly, we show that by transforming the problem into the phase domain, the original problem is converted to an approximated convex problem, which allows us to find a closed-form expression of a solution.
            Based on the identified solution, our analysis reveals the number of transmit antennas and the amount of time delay required for the best beam squint compensation. 
            \item 
            Leveraging the closed-form expressions of the proposed  approach, we formulate a mixed-integer optimization problem to determine the minimum number of TTDs required to achieve a given array gain performance. 
            Although the formulated mixed-integer problem is intractable, we show that by applying a second-order approximation, the original problem can be relaxed to a tractable form, enabling us to find a solution. 
            Our finding indicates that the number of TTDs linearly increases with respect to the system bandwidth to guarantee a required array gain performance.
            \item 
            We conduct extensive simulations to evaluate the performance of the proposed joint PS and TTD precoding. 
            The simulation results confirm that the beam squint is compensated effectively with our design. 
            {Simulations reveal that our joint optimization approach achieves superior array gain performance and computational efficiency compared to prior  TTD-based precoding methods.} 
            The simulations also verify the significant array gain performance improvement with the minimum number of TTDs informed by our optimization approaches.   
\end{itemize}
     \textbf{Synopsis}: The remainder of the paper is organized as follows. 
     Section~\ref{secII} presents the channel model of the wideband sub-THz massive MIMO OFDM system and analyzes the array gain loss caused by the beam squint.  
     Section~\ref{secIII} describes the relationship between beam squint compensation by the ideal analog precoder and the achievable rate. 
     Section~\ref{secIV} derives the closed-form solution of the TTD and PS precoders under the practical TTD constraints and determines the minimum number of TTDs that ensures a predefined array gain performance. 
     Section~\ref{secV} provides simulation results to validate the developed analysis.
    Section~\ref{secVI} concludes this work.\looseness=-1

    \textit{Notation:} A bold lower case letter $\bx$ is a column vector and a bold upper case letter $\bX$ is a matrix. 
    $\bX^T$, $\bX^H$, $\bX^{-1}$, $\|\bX\|_F$, $\tr(\bX)$, $\det(\bX)$, $\bX(i,j)$, $\|\bx\|_2$, and $|x|$ are, respectively, the transpose, conjugate transpose, inverse, Frobenius norm, trace, determinant, $i$th row and $j$th column entry of $\bX$, $2$-norm of $\bx$, and modulus of $x\in \C$. 
    $\blockdiag(\bx_1, \bx_2, \dots, \bx_N)$ is an $nN \times N$ block diagonal matrix such that its main-diagonal blocks contain $\bx_i \in \C^{n \times 1}$, for $i=1,\dots,N$ and all off-diagonal blocks are zero.
    $\mathbf{0}_n$, $\mathbf{1}_{n}$, and $\bI_{n}$ denote, respectively, the $n\times 1$ all-zero vector, $n\times 1$ all-one vector, and $n\times n$ identity matrix.
    Given $\bx \in \R^{n \times 1}$, $e^{j\bx}$ denotes the column vector $[e^{jx_1}~e^{jx_2} ~\dots ~e^{jx_n}]^T \in \C^{n \times 1}$ obtained by applying $e^j$ element-wise.
\section{Channel Model and Beam Squint Effect}
\label{secII}
In this section, we present the channel model of the wideband sub-THz massive MIMO OFDM system. Then, the array gain loss caused by the beam squint is analyzed. 
%%\vsapce{-0.2cm}
\subsection{Channel Model}
  We consider the downlink of a wideband sub-THz hybrid massive MIMO OFDM system where the transmitter is equipped with an $N_t$-element uniform linear array (ULA) with element spacing $d$.
 The transmit antenna array is fed by $N_{RF}$ radio frequency (RF) chains to simultaneously transmit $N_s$ data streams to an $N_r$-element ULA receiver. 
 It is assumed that $N_t, N_r, N_{RF},$ and $N_{s}$ satisfy $ N_s = N_{RF} \leq N_r \ll N_{t}$. 
 Herein, we let $f_c$, $B$, and $K$ be, respectively, the central (carrier) frequency, bandwidth of the OFDM system, and the number of OFDM subcarriers (an odd number).
 Then, the $k$th subcarrier frequency is given, \text{ for } $k = 1,\dots,K$, by
 %%vspace{-0.1cm}
 %\vsapce{-0.3cm}
%\small
 \begin{equation}
     \label{eqfk}
     f_k = f_c + \frac{B}{K}\Big(k-1-\frac{K-1}{2}\Big).
%\vsapce{-0.3cm}
 \end{equation}\normalsize
%%vspace{-0.4cm}
 %

 %
 The frequency domain MIMO OFDM channel at the $k$th subcarrier $\bH_k \in \C^{N_r \times N_t}$ is
 %%vspace{-0.1cm}
%\vsapce{-0.3cm}
%\small
\begin{equation}
     \label{eqHk}
     \bH_k = \sqrt{\frac{N_rN_t}{L}} \sum_{l=1}^{L} \alpha_{k,l} e^{-j 2\pi \tau_l f_k} \bu_{k,l}\bv^H_{k,l},
%\vsapce{-0.3cm}
 \end{equation}\normalsize 
 where $L$ denotes the number of channel (spatial) paths and $\tau_{l} \in \R$ represents the delay of the $l$th channel path. 
The $\alpha_{k,l} \in \C$ denotes the path gain of the $k$th subcarrier on the $l$th path, which incorporates the molecular absorption loss of sub-THz wave propagation medium.
The magnitude of $\alpha_{k,l}$ is modeled by $ \E [ |\alpha_{k,l}|^2 ] = \Big(\frac{c}{4 \pi f_k \widetilde{d}}\Big)^2 e^{-\cK_{\text{abs}}(f_k)\widetilde{d}}$, where $c= 3 \times 10^8$ m/s denotes the speed of light, $\widetilde{d}$ is the transmit distance,  and $\cK_{\text{abs}}(f)$ is the frequency-dependent medium absorption coefficient \cite{Han2015, Jornet2011}. 
{We note that the sub-THz channel is unlikely to have a number of spatial paths $L$ that is larger than the number of RF chains $N_{RF}$ \cite{sarieddeen2021}. 
    For ease of exposition, we assume $L = N_{RF}$, which is equivalent to setting $\alpha_{k,L+1} = \dots = \alpha_{k,N_{RF}} = 0$ when $L < N_{RF}$, for $k=1,2,\dots,K$.
    This setting has the same effect as using $L$ RF chains for transmission while turning off $(N_{RF}-L)$ RF chains.}
In the remainder of this paper, we use the subscript $l$ for denoting the index of both channel path and RF chain unless specified otherwise. 
 The vectors $\bv_{k,l} \in \C^{N_t \times 1}$ and $\bu_{k,l} \in \C^{N_r \times 1}$ in \eqref{eqHk} are the normalized transmit and receive array response vectors of the $k$th subcarrier on the $l$th path, respectively, where the $\bv_{k,l}$ is a function of angles of departure (AoDs) $\Psi_{l}\in [\Psi_{\min},\Psi_{\max}]$.
 % %%
 The exact same definition applies to the normalized receive array response vector $\bu_{k,l}$ in \eqref{eqHk}.\looseness=-1  
 
 In what follows, we will limit our discussion to the transmit antenna array, keeping in mind that the same applies to the receive array \cite{Han2022,dai2021,Gao2021,Matthaiou2021}. 
 The antenna geometry of the transmit array is described assuming far-field spatial angles.  
 The $n$th entry of $\bv_{k,l}\in \C^{N_t \times 1}$ in \eqref{eqHk} is $\bv_{k,l}(n,1) = \frac{1}{\sqrt{N_t}}e^{-j \pi \frac{2d f_k}{c} (n-1) \sin(\Psi_l)}$.
 We define $\psi_{k,l} = \frac{2d f_k}{c}\sin(\Psi_l)$ as the spatial direction of the $k$th subcarrier at the transmitter on the $l$th path. 
Assuming the half-wavelength antenna spacing, i.e., $d = \frac{c}{2f_c}$, the spatial direction at the central frequency is simplified to $\psi_l = \sin(\Psi_l)$.
Hence, setting $\zeta_k = \frac{f_k}{f_c}$ leads to 
$\psi_{k,l} = \zeta_k \psi_{l}$ and 
%\vsapce{-0.2cm}
%\small
\begin{equation}
\zeta_k = 1 + \frac{B}{f_c}\Big(\frac{k-1-\frac{K-1}{2}}{K}\Big) \label{eqSDc},
%\vsapce{-0.3cm}
 %vspace{-0.1cm}
\end{equation}\normalsize
where \eqref{eqSDc} follows from the definition of $f_k$ in \eqref{eqfk}. 
As a result, the transmit array response vector $\bv_{k,l}$ in \eqref{eqHk} is succintly expressed as $\bv_{k,l} = \frac{1}{\sqrt{N_t}}\Big[ 1~ e^{j\pi\psi_{k,l}} ~\dots ~e^{j\pi (N_t-1)\psi_{k,l}} \Big]^H  \in  \C^{N_t \times 1}$.\looseness=-1
% %%
%%\vsapce{-0.2cm}
\subsection{Beam Squint Effect}
     As aforementioned in Section~\ref{secI}, when the wideband OFDM is employed in a massive MIMO system, a substantial array gain loss at each subcarrier could occur due to beam squint. To quantify it, we focus on the $l$th path of the channel in \eqref{eqHk}. Denoting the frequency-independent beamforming vector matched to the transmit array response vector with the AoD $\Psi_l$ as $\bff^{(l)} = \bv_{c,l}$ (i.e., array response vector at the central frequency $f_c$ on the $l$th path), the corresponding array gain at the $k$th subcarrier is given by $g(\bff^{(l)},\psi_{k,l}) = |\bv^H_{k,l}\bff^{(l)}|$ {\cite{dai2021}, \cite{Wang2018}}, i.e.,  
     %\vsapce{-0.3cm}
     %\small
     \begin{equation}
     \label{eq5}
            \d4 g(\bff^{(l)},\psi_{k,l}) = \frac{1}{N_t}\bigg| \sum_{n=0}^{N_t-1}e^{j n \pi (\psi_{k,l} - \psi_{l})}\bigg|
           =\bigg| \frac{\sin(N_t \Delta_{k,l})}{N_t \sin(\Delta_{k,l})}\bigg|,    
     %\vsapce{-0.3cm}
     \end{equation}\normalsize 
    where $\Delta_{k,l} = \frac{\pi}{2}(\psi_{k,l}-\psi_{l})$.
    It is not difficult to observe that at the central frequency, the array gain is $g(\bff^{(l)},\psi_{k,l}) = 1$ because $\lim_{x \rightarrow 0}\frac{\sin(N_t x)}{\sin(x)} = N_t$.
    However, when $f_k \neq f_c$, $\Delta_{k,l}$ deviates from $0$; the amount of deviation increases as $f_k$ approaches to $f_1$ or $f_K$.
    As a result, all subcarriers except for the central subcarrier suffer from the array gain loss. 
    The implication in the spatial domain is that the beams at non-central subcarriers may completely split from the one generated at the central subcarrier.\looseness=-1 

           \begin{figure*}[htb]
        \centering
        \subfloat[]{
\includegraphics[scale=.34,trim=0.0cm 0cm 0.0cm 0.2cm, clip]{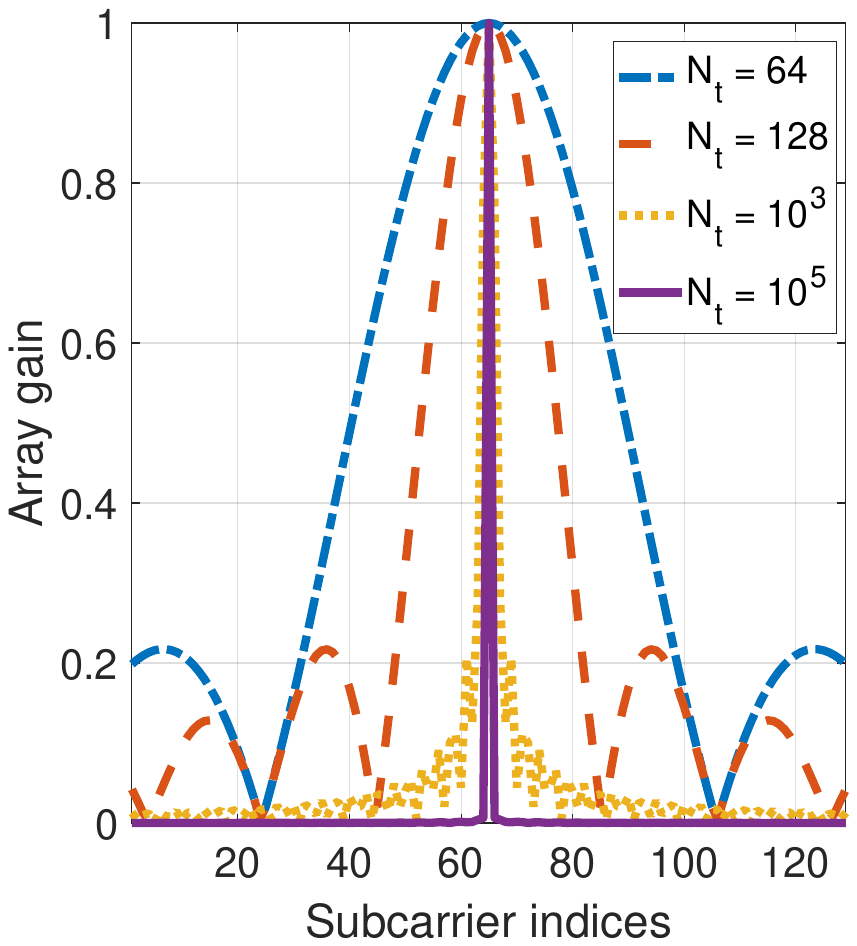}
        \label{fig1a}
        }
       \centering
        \subfloat[]{
    \includegraphics[scale=.34,trim=0cm 0cm 0.0cm 0.2cm, clip]{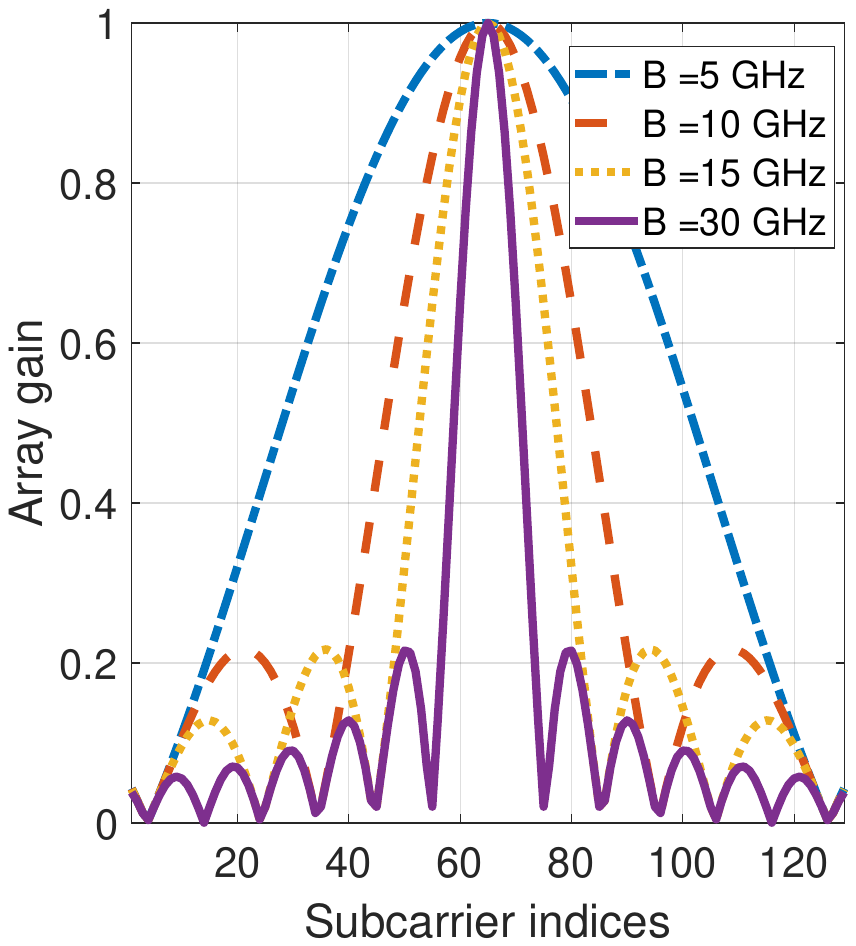}
\label{fig1b}
        }
   \centering
    \subfloat[]{
    \includegraphics[scale=.33,trim=0.0cm 0cm 0.0cm 0.2cm, clip]{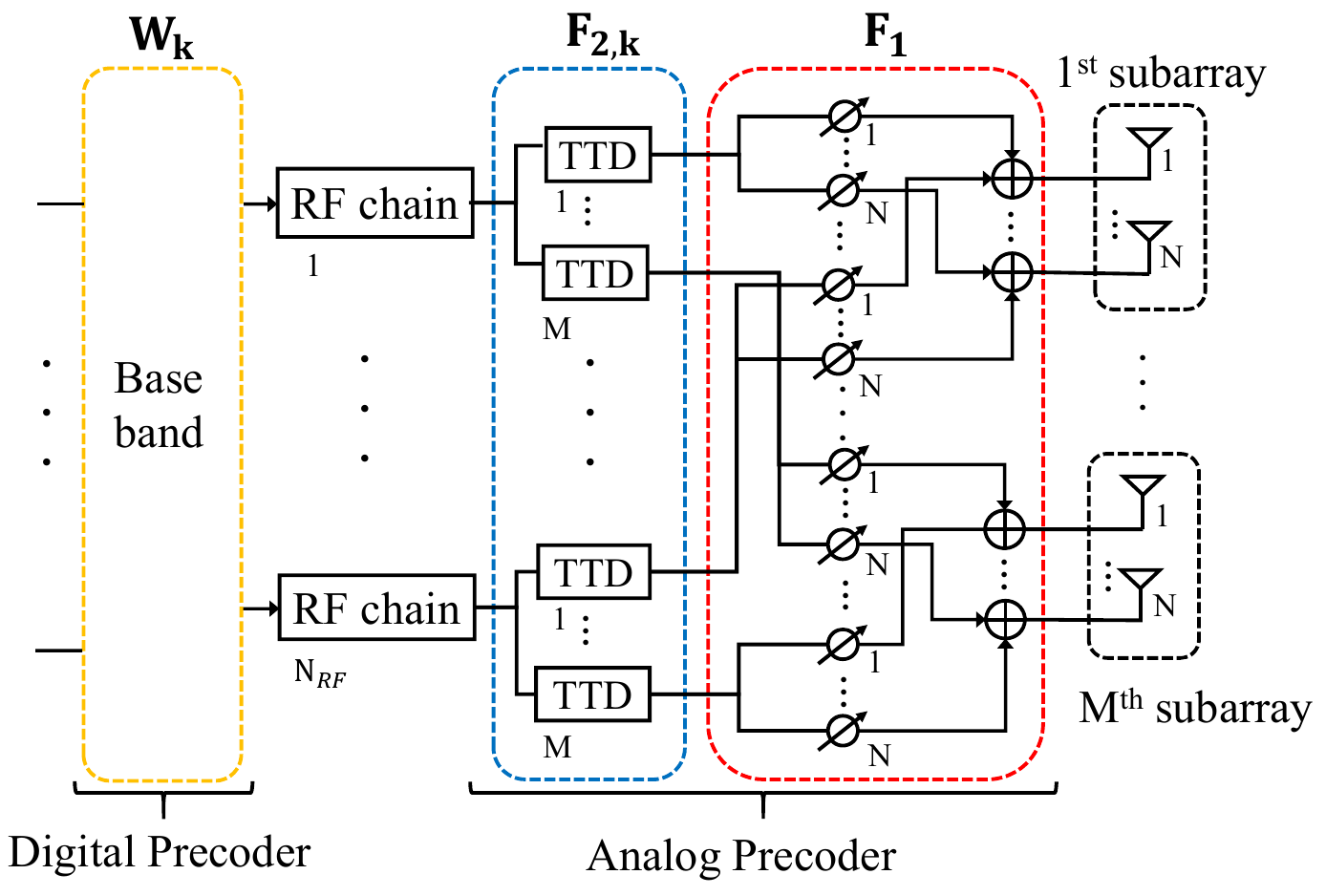}
         \label{fig2}}  
    %\vsapce{-0.5cm}
        \caption{(a) Array gain vs. subcarrier indices for different numbers of transmit antennas ($N_t$). (b) Array gain vs. subcarrier indices for different bandwidths ($B$). (c) TTD-based hybrid precoding architecture.}
        %\vspace{-0.5cm}    
    \label{fig1}
    \end{figure*}
 
 The following proposition quantifies the asymptotic array gain loss as $N_t \rightarrow \infty$.      
    \begin{proposition} 
    \label{ps1} 
    Suppose that $\psi_{k,l}$ is the spatial direction at the $k$th subcarrier ($f_k \neq f_c$) of the $l$th path. 
    Then, the array gain in \eqref{eq5} converges to $0$ as $N_t$ tends to infinity, i.e., $g(\bff^{(l)}, \psi_{k,l}) \stackrel{\boldsymbol{\cdot}}{=} 0$, where $\stackrel{\boldsymbol{\cdot}}{=}$ denotes the equality when $N_t \rightarrow \infty$.
    \end{proposition}
    \begin{proof}
    The  array gain in \eqref{eq5} can be rewritten as
  %  \begin{equation*}
      $ g(\bff^{(l)},\psi_{k,l}) = \frac{1}{N_t} \Big|\frac{\sin (N_t\Delta_{k,l})}{\pi\Delta_{k,l}} \frac{\pi \Delta_{k,l}}{\sin (\Delta_{k,l})} \Big|,$ 
    %\end{equation*} 
    where $\Delta_{k,l} \neq 0$ because $f_k \neq f_c$.
    The proposition follows from the definition of Dirac delta function $\frac{\sin{(N_t \Delta_{k,l})}}{\pi \Delta_{k,l}} \stackrel{\boldsymbol{\cdot}}{=} \delta(\Delta_{k,l})$ \cite{bracewell2000}, completing the proof. 
    \end{proof}

    %%
    % %%
    Fig.~\ref{fig1a} illustrates the convergence trend of {Proposition}~1. The array gain patterns are calculated for $f_c = 300$ GHz, $K = 129$, $\psi_{l}= 0.8$, and $B= 30$ GHz. 
    As $N_t$ tends to be large, the maximum array gain is only achieved at the central frequency, i.e., the $65$th subcarrier in Fig.~\ref{fig1a}, while other subcarriers suffer from substantial array gain losses. 
    This is quite the opposite of the traditional narrowband massive MIMO system in which the array gain grows as $N_t \rightarrow \infty$.
    The array gain loss is also numerically understood when the bandwidth $B$ increases. 
    Fig.~\ref{fig1b} illustrates those patterns while using the same parameters as in Fig.~\ref{fig1a} except for that $N_t=256$.
    As the bandwidth $B$ grows, the maximum array gain is 
    obtained only at the central frequency while other subcarriers experience a large amount of array gain losses. 
    This contrasts with the information-theoretic insight that the  achievable rate grows linearly with the bandwidth while keeping the signal-to-noise-ratio (SNR) fixed.
    Wideband sub-THz massive MIMO communication research is in its early stages. In order to truly unleash the potential of sub-THz communications, a hybrid precoding architecture and methods that can effectively compensate for the beam squint effect is of paramount importance. 
\section{Preliminaries and Achievable Rate Performance}
\label{secIII}
    We consider a TTD-based hybrid precoding architecture \cite{dai2021,Gao2021}, where each RF chain drives $M$ TTDs and each TTD is connected to $N$ PSs as shown in Fig. \ref{fig2}. 
   The $N_t$-element ULA is divided into $M$ subarrays with $N = \frac{N_t}{M}$ antennas per subarray.      
    The signal at the $k$th subcarrier passing through a TTD is delayed by $t$ ($0\leq t \leq t_{\max}$) in the time domain, which corresponds to the $-2\pi f_k t$ frequency-dependent phase rotation in the frequency domain.
    The $t_{\max}$ is the maximum time delay value that a TTD device can produce.  
    The $k$th subcarrier signal at the receiver $\by_k \in \C^{N_r \times 1}$ is then given by \cite{Qua2022} 
     \begin{equation}
         \label{eq6}
         \by_k = \sqrt{\rho} \bH_k\bF_1\bF_{2,k}(\{\bt_l\}^{N_{RF}}_{l=1})\bW_k\bs_k + \bn_k,
     \end{equation}
   where $\rho$, $\bs_k \in \C^{N_s \times 1}$, $\bW_k \in \C^{N_{RF} \times N_s}$, and $\bn_k \in \C^{N_r \times 1}$ are, respectively, the average transmit power, transmit data stream, baseband digital precoder, and normal Gaussian noise vector with each entry being independent and identically distributed (i.i.d.) according to zero mean and variance $1$.
    The $\bF_1$ in \eqref{eq6} is the PS precoding matrix, which is the concatenation of $N_{RF}$ PS submatrices $\bX_l \in \C^{N_t \times M}$, $1 \leq l \leq N_{RF}$, specifically,  
   %\vsapce{-0.3cm}
    %\small
    \begin{equation}
    \label{eqF1}
    \bF_1= \frac{1}{\sqrt{N_t}} [\bX_1 ~ \bX_2 ~ \dots ~  \bX_{N_{RF}}] \in \C^{N_t \times MN_{RF}},    
  %\vsapce{-0.3cm}
    \end{equation}\normalsize 
    where $\bX_l = \blockdiag(e^{j \pi \bx^{(l)}_1},e^{j \pi \bx^{(l)}_2},\dots,e^{j \pi \bx^{(l)}_{M}}) \in \C^{N_t \times M}$, $\bx^{(l)}_m = [x^{(l)}_{1,m}~x^{(l)}_{2,m}~\dots~x^{(l)}_{N,m}]^T\in \R^{N \times 1}$ is a PS vector, and $x^{(l)}_{n,m}$ is the value of the $n$th PS that is connected to the $m$th TTD on the $l$th RF chain, for $1 \leq n \leq N$, $ 1 \leq l \leq N_{RF}$, and $1 \leq m \leq M$. 
    The $\bF_{2,k}(\{\bt_l\}^{N_{RF}}_{l=1}) \in \C^{MN_{RF} \times N_{RF}}$ in \eqref{eq6} is the TTD precoding matrix, which is defined as
    %%
   %\vsapce{-0.3cm}
   %\small
    \begin{equation}
    \label{eqF2k}
     \d4\scalemath{0.9}{\bF_{2,k}(\{\bt_l\}^{N_{RF}}_{l=1}) = \blockdiag(e^{-j2\pi f_k \bt_1},e^{-j2\pi f_k \bt_2}, \dots, e^{-j2\pi f_k \bt_{N_{RF}}}),} 
    %\vsapce{-0.3cm}
    \end{equation}    %%
      where $\bt_l = [t^{(l)}_1~t^{(l)}_2~\dots~ t^{(l)}_{M}]^T \in \R^{M \times 1}$ is the $l$th time delay vector and $t^{(l)}_m$ is the time delay value of the $m$th TTD on the $l$th RF chain.
    For ease of exposition, in what follows, we omit the time delay vectors $\{\bt_l\}^{N_{RF}}_{l=1}$ in the TTD precoder notation $\bF_{2,k}(\{\bt_l\}^{N_{RF}}_{l=1})$, $\forall k$.
    We note that the analog precoder corresponds to the product $\bF_1\bF_{2,k} \in \C^{N_t \times N_{RF}}$ with the constant modulus constraint 
     %\vsapce{-0.3cm}
     %\small
    \begin{equation}
    \label{eqCM}
        \bF_1\bF_{2,k} \in \cF_{N_t,N_{RF}},
       %\vsapce{-0.3cm}
    \end{equation}\normalsize 
    where $\cF_{N_t,N_{RF}}$ denotes the set of all matrices $\bX \in \C^{N_t \times N_{RF}}$ such that $|\bX(i,j)| = \frac{1}{\sqrt{N_t}}, \forall  i, j$. 
    The data stream $\bs_k$ in \eqref{eq6} satisfies $\E[\bs_{k}\bs^{H}_{k}] = \frac{1}{N_s}\bI_{N_s}$.
    Then, the precoders are normalized such that $\|\bF_1\bF_{2,k}\bW_k\|^2_{F} = N_s$,  leading to $        \E[\|\bF_1\bF_{2,k}\bW_k\bs_k\|^2_2] = 1$, for $1 \leq k \leq K.$

    A conventional approach to maximize the system's achievable rate is to  jointly design the digital precoders $\{\bW_k\}$, TTD precoders $\{\bF_{2,k}\}$, and PS precoder $\bF_1$ in \eqref{eq6}. 
    However, this approach is intractable due to the coupling between variables and non-convex constraints of TTD and PS precoders \cite{ayach2014,Alkhateeb2014,JZhang2016}. 
    Therefore, we adopt a sequential hybrid precoding method 
    \cite{Alkhateeb20152, TJ2020, TJ2022}. 
    Focusing on the analog domain, we propose to jointly design the TTD and PS precoders to compensate for the beam squint while fixing the digital precoders. 
    Subsequently, the digital precoders can be optimized by standard approaches such as \cite{dai2021,Gao2021} to maximize the achievable rate. \looseness=-1
    \subsection{Preliminaries}
    \label{secIIIA}

   We first describe the sign invariance property of the array gain in \eqref{eq5} and then identify the ideal analog precoder that completely compensates for the beam squint. The sign invariance property and the ideal analog precoder established in this subsection will then be used in Section~\ref{secIV} for jointly optimizing TTD and PS precoders.   
  \subsubsection{Sign Invariance of Array Gain}
  \label{SignInvariance}
   We define the combination of PS precoder $\bF_1$ and TTD precoder $\bF_{2,k}$ as 
   %\vsapce{-0.4cm}
   \begin{equation}
    \label{eqFk}
       \bF_k = \bF_1\bF_{2,k}, \text{ for } 1 \leq k \leq K, 
    %\vsapce{-0.3cm}
   \end{equation}
    where the $l$th column of $\bF_k$ is $\bff^{(l)}_{k} = \frac{1}{\sqrt{N_t}}\bX_le^{-2\pi f_k \bt_l}$, for $1 \leq l \leq N_{RF}$. The array gain associated with $\bff^{(l)}_{k}$ is then given, based on \eqref{eq5}, by     
   %%
   %\vsapce{-0.4cm}
    \begin{equation}
    %\begin{multline}
    \label{eq7}
      \d4 g(\bff^{(l)}_{k},\psi_{k,l}) = \frac{1}{N_t}\bigg| \sum\limits_{m=1}^{M}\sum\limits_{n=1}^{N} e^{j \pi \zeta_{k} \gamma^{(l)}_{n,m}} e^{j \pi  x^{(l)}_{n,m}} e^{-j\pi \zeta_k  \vartheta^{(l)}_{m}}\bigg|,\!\! %\hspace{-0.2cm}
  %\vsapce{-0.3cm}
    %\end{multline}
    \end{equation} 
    where 
    \small$\gamma^{(l)}_{n,m} = ((m-1)N+n-1)\psi_{l}$\normalsize~ and \small$\vartheta^{(l)}_{m} = 2f_ct^{(l)}_m \in [0,\vartheta_{\max}]$\normalsize~ with \small$\vartheta_{\max} = 2f_c t_{\max}$\normalsize. 
    Note that in \eqref{eq7}, $e^{j \pi \zeta_{k} \gamma^{(l)}_{n,m}}$ is frequency-dependent while $e^{j \pi  x^{(l)}_{n,m}}$ is frequency-independent,   and $e^{-j\pi \zeta_k  \vartheta^{(l)}_{m}}$ depends on both time and frequency. 
    The beam squint compensation capability of the $m$th TTD on the $l$th RF chain is restricted because the phase rotation $-\pi \zeta_k  \vartheta^{(l)}_{m}$ is within the interval $[-\pi \zeta_k  \vartheta_{\max},0]$.\looseness=-1 
    \begin{remark} (Sign Invariance Property)
    \label{rmk1}
    The $g(\bff^{(l)}_{k},\psi_{k,l})$ in \eqref{eq7} is invariant to the multiplication of negative signs to $\gamma_{n,m}^{(l)}$, $x_{n,m}^{(l)}$, and $\vartheta^{(l)}_{m}$. 
    To be specific, given $\psi_{l} \geq 0$ (i.e., $\gamma^{(l)}_{n,m} \geq 0$, $\forall m$, $n$), we denote $\{ {x_{n,m}^{(l)}}^\star \}$ and $\{{\vartheta^{(l)}_{m}}^\star\}$ as the optimal values of $\{ x_{n,m}^{(l)}\}$ and $\{ \vartheta^{(l)}_{m}\}$, respectively, that maximize $g(\bff^{(l)}_{k},\psi_{k,l})$ in \eqref{eq7}. 
        Then, it is not difficult to observe that $\{-{x^{(l)}_{n,m}}^\star\}$ and $\{\vartheta_{\max} - {\vartheta^{(l)}_{m}}^\star\}$ also maximize $g(\bff^{(l)}_{k},-\psi_{k,l})$. 
        Hence, without loss of generality, in what follows, we assume that  
        %\vsapce{-0.5cm}
        \begin{equation}
        \label{eqSIP}
            \psi_{l} \geq 0, \text{ for } 1 \leq l \leq N_{RF}.
            %\vsapce{-0.5cm}
        \end{equation}
         \noindent The sign invariance property will be found to be useful when deriving the solution to our optimization problem for joint time delay and phase shift precoding in Section~\ref{secIV}.   
    \end{remark} 
  \subsubsection{Ideal Analog Precoder}
  An ideal analog precoder can produce arbitrary frequency-dependent phase rotation values to completely mitigate the beam squint effect, which is hereafter denoted as $\widetilde{\bF}^{\star}_k \in \cF_{N_t \times N_{RF}}$, for $1 \leq k \leq K$. 
  The purpose of invoking the ideal analog precoder is to provide a reference design for the proposed joint TTD and PS precoding method in Section~\ref{secIV}. 
  Denoting $\widetilde{\bff}^{(l)}_k$ as the $l$th column of $\widetilde{\bF}^{\star}_k$, the array gain obtained by $\widetilde{\bff}^{(l)}_k$ is $g(\widetilde{\bff}^{(l)}_k,\psi_{k,l}) = |\bv^H_{k,l}\widetilde{\bff}^{(l)}_k| = 1$, for $1 \leq l \leq N_{RF}$,  
  resulting in $\widetilde{\bF}^{\star}_k = [\bv_{k,1}~\bv_{k,2}~ \dots~\bv_{k,N_{RF}}]$. Thus, the $((m-1)N\+n)$th row and $l$th column entry of $\widetilde{\bF}^{\star}_k$ is \looseness=-1
    %\vsapce{-0.3cm}
    \begin{equation}
        \label{eqIdealAP}
        \widetilde{\bF}^{\star}_k((m-1)N+n,l) =\frac{1}{\sqrt{N_t}} e^{-j\pi\zeta_k\gamma^{(l)}_{n,m}}, \forall k, l,m, n.
    %\vsapce{-0.3cm}
    \end{equation} 
The ideal analog precoder $\widetilde{\bF}^{\star}_k$ in \eqref{eqIdealAP} is achievable with $\bF_1\bF_{2,k}$ when each transmits antenna is equipped with a dedicated TTD, i.e., $M = N_t$ and $N =1$. In this case, the $l$th column of $\bF_1\bF_{2,k}$ is matched exactly to the array response vector of the $k$th subcarrier on the $l$th path, $1 \leq l \leq N_{RF}$ and $1 \leq k \leq K$. 
For instance, the PS and TTD values are designed as $x^{(l)}_{1,m} = 0$ and $t^{(l)}_m = \frac {(m-1)\psi_l}{2f_c}$, respectively, $1 \leq m \leq N_t$ and $ 1 \leq l \leq N_{RF}$.
However, this design is not practically motivated because it requires $N^2_t$ TTDs, resulting in a high hardware complexity and huge power consumption. 
In Section~\ref{secIV}, we address the latter issue by proposing a method to minimize the power consumption of the analog precoder given an array gain performance requirement.  
\subsection{Achievable Rate Performance}  
\label{secIIIB}
    {Given the analog precoders $\{\bF_k\}_{k=1}^{K}$ in \eqref{eqFk}, the achievable rate (averaged over the subcarriers) of the channel in \eqref{eqHk} is
    \begin{equation}\label{eqRate}
    R = \frac{1}{K} \sumK \log_2 \det  \Big( \bI_{N_s}  + \frac{\rho}{N_s} \bH_k\bF_k\bW_k\bW^H_k\bF^H_{k}\bH^H_k \Big).
\end{equation}To quantify the impact of the ideal analog precoders $\{\widetilde{\bF}^{\star}_k\}$ on the achievable rate $R$ in \eqref{eqRate}, we define the singular value decomposition (SVD) of $\bH_k$ as $\bH_k = \bU_k\bSig_k\bV^H_k$, where $\bU_k \in \C^{N_r \times N_{RF}}$ satisfying $\bU^H_k\bU_k = \bI_{N_{RF}}$, $\bV_k \in \C^{N_t \times N_{RF}}$ satisfying $\bV^H_k\bV_k = \bI_{N_{RF}}$, and $\bSig_k \in \R^{N_{RF} \times N_{RF}}$ is a diagonal matrix of non-zero singular values arranged in descending order.
It is well understood that the optimal fully-digital precoders maximizing $R$ in \eqref{eqRate} are $\{\bV_k\}$ \cite{Hadi2016}. 
Consequently, the precoders $\{\bF_k\}$ and $\{\bW_k\}$ that achieve this maximal rate can be designed by solving $\bF_k \bW_k = \bV_k$, for $k=1,\dots, K$.} 
   \begin{figure*}[htb]
        \centering
        \subfloat[]{%
        \includegraphics[scale=.6,trim=0.0cm 0cm 0.0cm 0.2cm, clip]{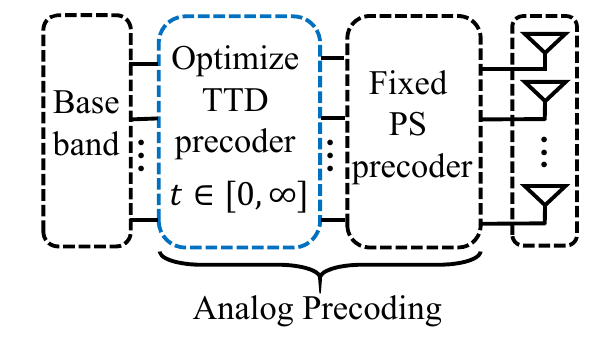}%
        \label{fig0a}%
        }%
        \hspace{5mm}
        \centering
        \subfloat[]{%
        \includegraphics[scale=.6,trim=0.0cm 0cm 0.0cm 0.2cm, clip]{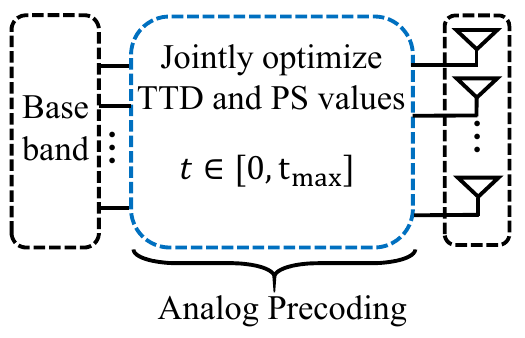}%
        \label{fig0b}%
        }%
       \hspace{5mm}
       \centering
        \subfloat[]{%
    \includegraphics[scale=.2,trim=0cm 0cm 0.0cm 0.0cm, clip]{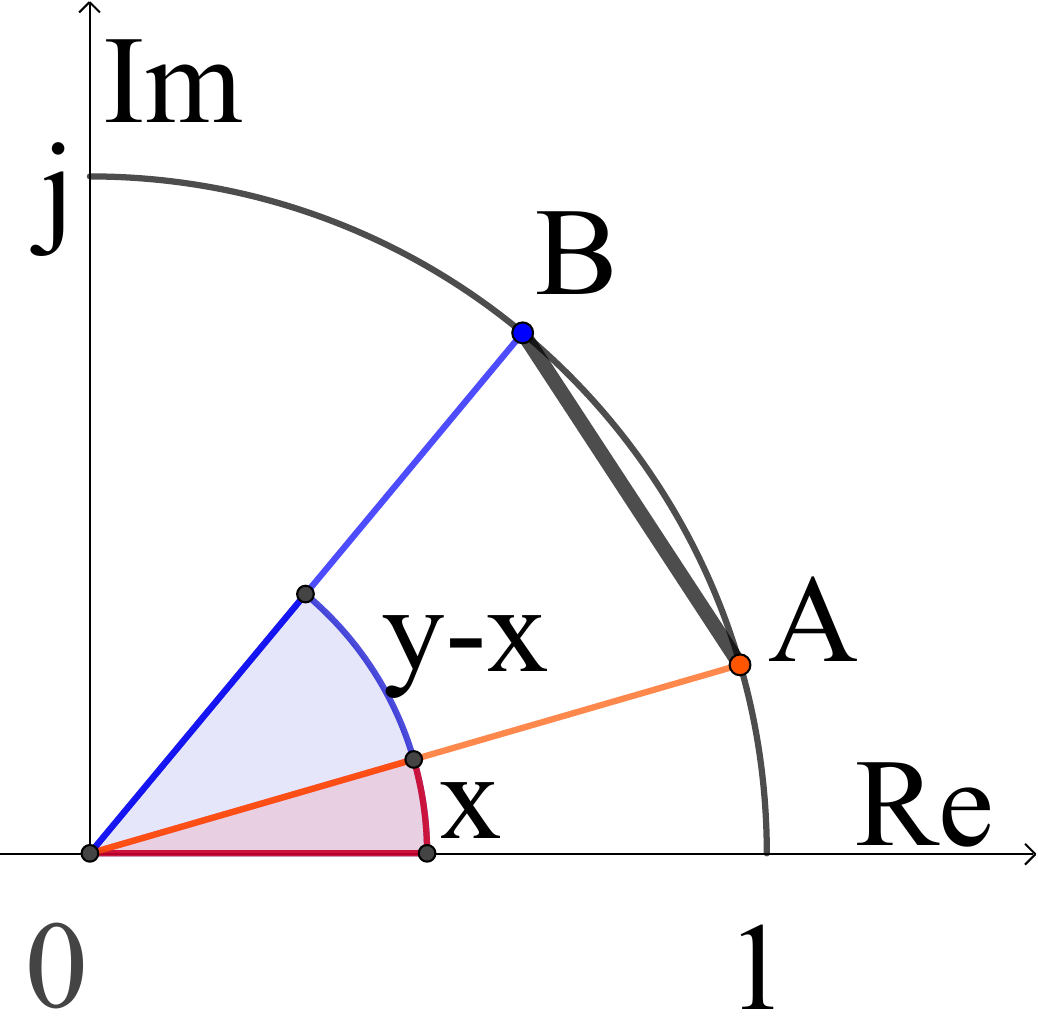}%  
    \label{fig4}%
        }%
    %%\vsapce{-0.5cm}    
     \caption{(a) Conceptual illustration of the prior methods for beam squint compensation \cite{dai2021,Gao2021,Matthaiou2021} . (b) The proposed method of this work. (c) Conceptual illustration showing the relationship between the distance of two points on the unit modulus circle (the segment AB) and their relative phase difference (the arc corresponding to $y-x$) }
    %%\vsapce{-0.8cm}
    \end{figure*}
{In subsequent analysis, we confirm the existence of a digital precoder $\bW^{\star}_k$ such that $\widetilde{\bF}_k^{\star} \bW_k^{\star} = \bV_k$, for $k=1,\dots, K$.
From the channel model in \eqref{eqHk}, we have $\bU_k\bSig_k\bV^H_k = \bG_k\bLam_k\widetilde{\bF}^{\star H}_k$, where $\bG_k = [\bu_{k,1} \dots \bu_{k,L}] \in \C^{N_r \times L}$  and $\bLam_k = \sqrt{\frac{N_tN_r}{L}} \diag(\alpha_{k,1}e^{-j2\pi\tau_lf_k},\dots,\alpha_{k,L}e^{-j2\pi\tau_lf_k})$.
It can be verified that $\bV_k = \widetilde{\bF}^{\star}_k\bLam_k^H\bG_k^H\bU_k(\bSig_k^T)^{-1}$ due to the assumption $L = N_{RF}$. 
It is feasible to match the performance of the optimal fully-digital precoder $\bV_k$ using the idea analog precoder $\widetilde{\bF}^{\star}_k$ combined with the digital precoder $\bW_k = \bLam_k^H\bG_k^H\bU_k(\bSig_k^T)^{-1}$. 
Therefore, we attempt in Section~\ref{secIV} to design the TTD precoder $\bF_{2,k}$ and PS precoder $\bF_1$ in order to best approximate the ideal analog precoder $\widetilde{\bF}^{\star}_k$, for $k=1,\dots, K$, under the per-TTD time delay constraint.}
    \section{Joint Delay and Phase Precoding  Under TTD Constraints} \label{secIV}
     We propose a novel method to design TTD and PS precoders, which is illustrated in Fig.~\ref{fig0b}. 
     In contrary with the prior works \cite{dai2021,Gao2021,Matthaiou2021} that optimize TTD values while fixing the PS values (Fig.~\ref{fig0a}), we jointly optimize both TTD and PS values (Fig.~\ref{fig0b}).
     Instead of assuming the impractical TTD that produces any time delay value $t$ (i.e., $0 \leq t \leq \infty$, Fig.~\ref{fig0a}), a finite interval  time delay value (i.e.,  $0 \leq t \leq  t_{\max}$, Fig.~\ref{fig0b}) is taken into account in our approach.  
    \subsection{Problem Formulation}
     \label{secIVA}
     Ideally speaking, we wish to find $\{\bt_l\}^{N_{RF}}_{l=1}$ and $\bF_1$  satisfying $\bF_1\bF_{2,k} = \widetilde{\bF}^{\star}_k$, $\forall k$. 
     However, given fixed $l$, $m$, and $n$, solving $K$-coupled matrix equations is an ill-posed problem. 
     To overcome this, we approach to formulate a problem that optimizes $\bF_1$ and $\{\bt_l\}^{N_{RF}}_{l=1}$ by minimizing the difference between $\bF_1\bF_{2,k}$ and $\widetilde{\bF}^{\star}_{k}$:  
     %\vsapce{-0.3cm}
     %\small
       \begin{subequations}
        \label{opt_prob}
        \beq
        \d4\d4\min_{\bF_1, \{ \bF_{2,k}\}_{k=1}^{K}} \d4 && \d4  \frac{1}{K}\sum_{k=1}^{K} \left\|\widetilde{\bF}^{\star}_k - \bF_1\bF_{2,k}\right\|^2_F \label{obj1}\\      
        \d4\d4\d4\d4\text{subject to}  
        \d4 && \d4 0 \leq t^{(l)}_{m} \leq t_{\max},\forall l,m,  \label{opt_proba}\\
        \d4 && \d4 {|\bF_1(i,j)|  \in  \left\{ \frac{1}{\sqrt{N_t}}, 0  \right\}, \forall i,j,} \label{opt_probb}\\
        \d4&&\d4 {|\bF_{2,k}(p,q)| \in  \{0,1\},\forall p,q,}\label{opt_probc} \\  
        \d4&& \d4{\bF_1\bF_{2,k}  \in  \cF_{N_t, N_{RF}},\forall k,}\label{opt_probd}
        \eeq
   \end{subequations}\normalsize 
   where the constraint in \eqref{opt_proba} indicates the restricted range of the time delay values per-TTD device, the constraint in \eqref{opt_probb} is due to the definition of $\bF_1$ in \eqref{eqF1}, the constraint in \eqref{opt_probc} is due to the definition of $\bF_{2,k}$ in \eqref{eqF2k}, and the constraint in \eqref{opt_probd} describes the constant modulus property of the analog precoder in \eqref{eqCM}.   
    The constraints in \eqref{opt_probb}-\eqref{opt_probd} in conjunction with the coupling between $\bF_1$ and $\bF_{2,k}$ in \eqref{obj1} make the problem difficult to solve. 
      Besides, \eqref{opt_prob} can be viewed as a matrix factorization problem with non-convex constraints, which has been studied in the context of hybrid analog-digital precoding \cite{ayach2014,Hadi2016,Zhang2018,Sohrabi2017}. 
     A common approach was applying block coordinate descent (BCD) and relaxing the constraints to deal with the non-convexity \cite{ayach2014,Hadi2016,Zhang2018,Sohrabi2017}. 
    Unlike the prior works, the original non-convex problem in \eqref{opt_prob} is approached, in this work, as an approximated convex problem. 
    \begin{remark}
    { We acknowledge the resemblance of the problem in \eqref{opt_prob} to the one studied in \cite{Ratnam2022}, which explores a generic beamforming behavior for fast frequency multiplexing. 
   The problem in [43, Eqn. (3)] can be viewed as a regularized variant of the problem in \eqref{opt_prob}. 
   Unlike the iterative method in \cite{Ratnam2022}, our work specifically targets the challenges of beam squint compensation, introducing a novel, one-shot joint design for PS and TTD values.
    This advancement allows us to delineate the system parameters, including $N_t, M, B,$ and $t_{\max}$ for effective beam squint compensation, which will be addressed in the later part of this section, a perspective that remains unexplored in the prior work \cite{Ratnam2022}.}
    \end{remark}
     Based on the structure of $\bF_1$ in \eqref{eqF1} and the block diagonal structure of $\{\bF_{2,k}\}_{k=1}^{K}$ in \eqref{eqF2k}, the objective function in \eqref{obj1} can be rewritten as
%    %\vsapce{-0.2cm}
    \begin{equation}
    \label{eqObj}  
     \frac{1}{K}\frac{1}{N_t}\sum_{k=1}^{K} \sum_{l=1}^{N_{RF}} \sum_{m = 1}^{M} \sum_{n=1}^{N} \Big|e^{-j\pi \zeta_k \gamma^{(l)}_{n,m}} - e^{j \pi x^{(l)}_{n,m}} e^{-j\pi \zeta_k \vartheta^{(l)}_{m} }\Big|^2.  
%    %\vsapce{-0.2cm}
    \end{equation}
    It is still difficult to deal with the objective function in \eqref{eqObj} due to the unit modulus constraint. 
    %%%
    To streamline the optimization process, we introduce a lemma demonstrating the equivalence of optimization on the unit circle to optimization in the phase domain.
     \begin{lemma}%%
     \label{lm1}
      For $x \in \R$ and $y \in \R$, the following equality holds 
      $\argmin\limits_{y \text{:} \mod(y,\pi) \neq x} |e^{jx} -e^{jy}| = \argmin\limits_{y \text{:} \mod(y,\pi) \neq x}|x-y|$,
         where $\mod(y,\pi)$ is $y$ modulo $\pi$.  
    %\vsapce{-0.1cm}
    \end{lemma}
    \begin{proof}
    See Appendix \ref{AppendixB}. 
    %\vsapce{-0.3cm}
    \end{proof}
   Fig.~\ref{fig4} graphically visualizes the equivalence in Lemma~\ref{lm1}. We let points A and B represent $e^{jx}$ and $e^{jy}$, respectively. It is not difficult to observe from Fig.~\ref{fig4} that minimizing the length of the segment AB with respect to B is equivalent to minimizing $|y-x|$ with respect to $y$, in which the latter is a convex optimization problem.
     In what follows, Lemma~\ref{lm1} is exploited to convert the non-convex problem in \eqref{opt_prob} into an approximated convex problem. 
      {
      Incorporating Lemma 1 into \eqref{eqObj} converts the problem in \eqref{opt_prob} to an approximation \cite{Qua2022}, \cite{Ratnam2022}:}
    %\small
    %\vsapce{-0.3cm}
     \begin{subequations}
             \label{opt_pro_2}    
        \beq
         \d4\min_{\{x^{(l)}_{n,m}\}, \{\vartheta^{(l)}_{m}\}} &&\d4\d4\scalemath{0.9}{\frac{1}{K}\!\sum\limits_{k=1}^{K}\!\sum\limits_{l=1}^{N_{RF}}\! \sum\limits_{m=1}^{M}\!\sum\limits_{n=1}^{N}\left|
          x^{(l)}_{n,m}  \!-\!  \zeta_k\vartheta^{(l)}_{m} \!+\! \zeta_k\gamma^{(l)}_{n,m}\right|^2}\!\!\!,          \label{opt_pro_2a}\\   
        \text{subject to }   &&\d4\d4 0 \leq \vartheta^{(l)}_{m}  \leq \vartheta_{\max}, ~\forall  l, m.         \label{opt_pro_2b}
        \eeq 
    \end{subequations}
     Next, we turn \eqref{opt_pro_2} to a composite matrix optimization problem. 
      The PS and TTD variables in \eqref{opt_pro_2} can be collected into a matrix \small$\bA_l = [\ba^{(l)}_{1}~\dots~\ba^{(l)}_{M}] \in \R^{(N+1)\times M}$\normalsize, where \small$\ba^{(l)}_{m} = [(\bx^{(l)}_m) ^T~\vartheta^{(l)}_{m}]^T \in \R^{N+1}$\normalsize.
      Containing the analog counterpart in $\widetilde{\bF}^{\star}_k$ in a matrix  $\bB^{(l)}_k \in \R^{N \times M},$ where the $n$th row and $m$th column entry of $\bB^{(l)}_k$ is $B^{(l)}_k(n,m) = -\zeta_k \gamma^{(l)}_{n,m}$, $\forall k,l,n,m$, the problem \eqref{opt_pro_2} becomes 
     %\vsapce{-0.3cm}
    %\small
    \begin{subequations}
    \label{opt_pro3}
        \beq
        \min_{\{ \bA_l \}^{N_{RF}}_{l=1}} && \frac{1}{K}\sum\limits_{k=1}^{K}\sum\limits_{l=1}^{N_{RF}} \left\|\bC_k\bA_{l} - \bB^{(l)}_k\right\|^2_F,\\  \text{ subject to }&&  
        \mathbf{0}^T_{M} \leq \be^{T}_{N+1}\bA_{l} \leq \vartheta_{\max}\mathbf{1}^T_{M}, ~\forall l,
     \eeq 
    \end{subequations}where \small$\bC_k = [\bI_{N} ~ -\zeta_k \mathbf{1}_{N}] \in \R^{N \times (N+1)}$\normalsize~ and $\be_{N+1} \in \R^{N+1}$ is the $(N+1)$th column of the {identity} matrix $\bI_{N+1}$. The vector inequalities in \eqref{opt_pro3} is the entry-wise inequalities. 
    
By introducing $\bC$ $=$ $\frac{1}{K}\sum_{k=1}^{K}\bC^T_k\bC_k$ $\in$ $\R^{(N+1)\times(N+1)}$, $\bD_l$ $=$ $\frac{1}{K}\sum_{k=1}^{K}\bC^T_k\bB^{(l)}_k$ $\in$ $\R^{(N+1)\times M}$\normalsize, and $c^{(l)}_B$ $=$ $\frac{1}{K}\sumK\|\bB^{(l)}_k\|^2_F$\normalsize, the objective function in 
     \eqref{opt_pro3} can be rewritten as 
$\frac{1}{K}\sum_{k=1}^{K}\sum_{l=1}^{N_{RF}} \Big\|\bC_k\bA_{l} - \bB^{(l)}_k\Big\|^2_F$
 % = \sum_{l=1}^{N_{RF}} \tr(\bC\bA_l\bA^T_l)-2\tr(\bA^T_l\bD_l) + c^{(l)}_B, \nonumber \\               
 $=$ $\sum_{l=1}^{N_{RF}}\sumM \Big(\big(\ba^{(l)}_m \big)^T \bC \ba^{(l)}_m - 2 \big(\bd^{(l)}_m\big)^T\ba^{(l)}_m \Big)+ \sum_{l=1}^{N_{RF}} c^{(l)}_B$\normalsize,
where the $\bd^{(l)}_m$ is the $m$th column of $\bD_l$. Hence, the problem \eqref{opt_pro3} is equivalently   
%    %\vsapce{-0.5cm}
  \begin{subequations}
        \label{eq:lcqp}   
        \beq
        \min_{\ba^{(l)}_m} &&(\ba^{(l)}_m)^T\bC\ba^{(l)}_m -2(\bd^{(l)}_m)^T\ba^{(l)}_m, \\
        \text{ subject to } 
         && 0 \leq \be_{N+1}^{T}\ba^{(l)}_m \leq \vartheta_{\max}, \forall l,m. 
        \eeq
    \end{subequations}
    %%\vsapce{-0.3cm}
    %%
    The problem \eqref{eq:lcqp} can be viewed as a decomposition of \eqref{opt_pro3} into $MN_{RF}$ independent problems. 
    This decomposition reveals an alignment with the TTD and PS precoding architecture in Fig.~\ref{fig2}, where each RF chain feeds $M$ TTDs and each TTD feeds $N$ PSs.  
    Thus, the problem in \eqref{eq:lcqp} is equivalent to optimizing the TTD and PS values of each RF chain branch independently.
    \subsection{Optimal Closed-form Solution}
    \label{secIVB}
    In this subsection, we find the optimal closed-form solution of the convex problem \eqref{eq:lcqp}.
    We start by deriving the expression of $\bC$ in \eqref{eq:lcqp}. To this end, we first define the constant $\Gamma = N+\eta$, where $\eta = \frac{NB^2}{f^2_c}\frac{(K^2-1)}{12K^2}$. For $\bC_k$ in \eqref{opt_pro3}, we have $\bC^T_k\bC_k =  \begin{bsmallmatrix}
     \bI_{N} & -\zeta_k \mathbf{1}_{N}  \\-\zeta_k \mathbf{1}^T_{N} & N\zeta^2_k
    \end{bsmallmatrix}$, where $\zeta_k$ is defined in \eqref{eqSDc}. 
    After some algebraic manipulations, it is readily verified that 
   %\vsapce{-0.3cm}
    %\small
    \begin{equation}
    \label{eqZeta}
    \sumK \zeta_k  = K \text{ and } \sumK \zeta^2_k  =K\Big(1+\frac{\eta}{N}\Big).
    %\vsapce{-0.3cm}
    \end{equation}\normalsize 
    Then, $~~\bC~~ = ~~~~ \frac{1}{K} \sumK \bC^T_k\bC_k$\normalsize~ is simplified, based on \eqref{eqZeta}, to 
    $\bC =  \begin{bsmallmatrix} 
        \bI_{N} & -\mathbf{1}_{N} \\
        -\mathbf{1}^T_{N} & \Gamma \end{bsmallmatrix}$\normalsize,  
    which yields \looseness=-1   
    %\vsapce{-0.5cm}
    \begin{equation}
        \label{eqCe}
         \be^T_{N+1}\bC^{-1}\be_{N+1} = \frac{1}{\eta}. 
    %\vsapce{-0.3cm}
    \end{equation}
    Defining $\bb^{(l)}_{k,m} \= [-\zeta_k\gamma^{(l)}_{1,m}~\dots~-\zeta_k\gamma^{(l)}_{N,m}]^T$ as the $m$th column of $\bB^{(l)}_k$ in \eqref{opt_pro3}, we attain \looseness=-1
    %\vsapce{-0.3cm}
    %\small
    \begin{equation}
        \label{eqdlm}
        \bd^{(l)}_m  = \begin{bmatrix}
      \frac{1}{K}\sumK \bb^{(l)}_{k,m} \\ \frac{1}{K}\sumK -\zeta_k \mathbf{1}^T_{N}\bb^{(l)}_{k,m}
     \end{bmatrix} \in \R^{(N+1)\times 1}.
    %\vsapce{-0.3cm}
    \end{equation}\normalsize
    Using \eqref{eqdlm}, it is readily verified that $\be_{N+1}^T\bC^{-1}\bd^{(l)}_{m} =\frac{1}{\eta} \frac{1}{K}\sumK \Big((1-\zeta_k)\sumN \bb^{(l)}_{k,m}(n,1) \Big)$. Hence,
    \begin{equation}\label{eqCde}
     \be_{N+1}^T\bC^{-1}\bd^{(l)}_{m}=\frac{(2m-1)N-1}{2}\psi_{l}, 
    %\vsapce{-0.3cm}
    \end{equation}\normalsize
    where the $\bb^{(l)}_{k,m}(n,1)$ is the $n$th entry of $\bb^{(l)}_{k,m}$, $1 \leq n \leq N$, and the last equality in \eqref{eqCde} follows from the facts that $\bb^{(l)}_{k,m}(n,1)$ $=$ $-\zeta_k \gamma^{(l)}_{n,m}$ and $\gamma^{(l)}_{n,m}$ $=$ $((m-1)N+n-1)\psi_l$. 
    Based on \eqref{eqCe} and \eqref{eqCde}, the optimal solution of the problem in \eqref{eq:lcqp} is summarized below.\looseness=-1  
    \begin{theorem}
    \label{theorem1}
    The optimal solution ${\ba^{(l)}_m}^\star = [({\bx^{(l)}_m}^\star)^T,{\vartheta^{(l)}_m}^\star]^T$ to \eqref{eq:lcqp} is given by 
    %%
    %\small
    \begin{subnumcases}{\label{eqPS} {x^{(l)}_{n,m}}^{\footnotesize \star}=}
    \scalemath{0.9}{\frac{N-2n+1}{2}\psi_{l}},&  \text{if } $0 \leq \psi_{l} \leq \frac{4f_c t_{\max}}{(2m-1)N-1} $, \label{eqPSa}\\ 
     \scalemath{0.9}{\vartheta_{\max} - \gamma^{(l)}_{n,m}},& $\text{ otherwise}$, \label{eqPSb}
    \end{subnumcases}
    for $1 \leq l \leq N_{RF}$, $1 \leq n \leq N$, $1 \leq m \leq M$,
    and ${\vartheta^{(l)}_m}^\star = 2f_c {t^{(l)}_m}^\star$, where the ${t^{(l)}_m}^\star$ is   
    %\small
    \begin{subnumcases}{\label{eqTTD}
    {t^{(l)}_m}^\star = } \scalemath{0.9}{\frac{(2m-1)N-1}{4f_c}\psi_{l}}, \d4 &$\text{ if } 0 \leq \psi_{l} \leq \frac{4f_c t_{\max}}{(2m-1)N-1}$,  \label{eqTTDa} \\  
    t_{\max},  \d4& $\text{ otherwise}$. \label{eqTTDb}   \end{subnumcases}
    \end{theorem}
    \begin{proof}
     See Appendix~\ref{AppendixTheorem1}. 
    %\vsapce{-0.3cm}
    \end{proof}
    \begin{remark}
    \label{rmk2}
    The solutions in \eqref{eqTTD} and \eqref{eqPS}  differentiate them from prior approaches to solving related problems. For instance, the PS values in \cite{dai2021} were not optimized and given by $x^{(l)}_{n,m} = -(n-1)\psi_{l}$, $\forall n$. 
    The time delay value of the $m$th TTD was $t^{(l)}_{m} = m\frac{N\psi_{l}}{2f_c}$ in \cite{dai2021}; as $m$ increases, the $t^{(l)}_m$ could be larger than $t_{\max}$, in which case such $t^{(l)}_m$ needs to be floored to the $t_{\max}$, resulting in performance deterioration.
     As will be discussed in {Section}~\ref{secV}, when all TTD values are smaller than $t_{\max}$, the approaches in \cite{dai2021,Gao2021} achieve the same array gain performance as the proposed approach, meaning that the designs in \cite{dai2021,Gao2021} is a special case of {Theorem}~\ref{theorem1}.
\end{remark}
  \begin{remark}
  \label{rmk3.1}
      The proposed approach takes advantage of the block-diagonal structure of PS and TTD precoders to decompose the MIMO precoding problem in \eqref{opt_prob} into multiple single-RF chain precoding design problems in \eqref{eq:lcqp}.
  For an array architecture that does not allow introducing a tractable block-diagonal structure of the precoders, a general subspace decomposition approach for the hybrid precoders design in \cite{ayach2014, Hadi2016} can be extended to alternatively optimize the digital precoder $\{\bW^{\star}_k\}$, PS precoder $\bF_1^\star$, and TTD precoder $\{\bF^{\star}_{2,k}\}$. 
    However, we note that satisfying the constraint in \eqref{opt_probd} without imposing some tractable structure into $\bF_1$ and $\{\bF_{2,k}\}$ is still challenging. 
    Investigation of ways of addressing general hybrid MIMO precoding architecture is subject to further research and not the focus of the present work.
\end{remark}
    In the following, the benefits of the proposed joint TTD and PS precoder optimization method are discussed, characterizing system parameters for the best beam squint compensation, which was unexplored by the prior works \cite{dai2021, Ratnam2022}. 
   The optimal condition in \eqref{eqPSa} and \eqref{eqTTDa} can be  rewritten as 
    \begin{equation}
    \label{eqcondi}
     \frac{(2m-1)N-1}{4f_c} \psi_l \leq t_{\max}.   
    \end{equation} 
    When \eqref{eqcondi} is satisfied, the beam squint is compensated effectively. 
    {In this case, the squinted beams are shifted to be aligned with spatial paths.} 
    However, the condition \eqref{eqcondi} may be violated, for example, when either $t_{\max}$ becomes small or $N$ becomes large (i.e., when $N_t$ tends to be large while fixing $M$, i.e., $N = \frac{N_t}{M}$).
    {This case results in beam misalignment at some subcarrier frequencies, which could cause array gain degradation.} 
    Based on the condition in \eqref{eqcondi}, we obtain selection criteria (rule of thumb) on the required number of transmit antennas $(N_t)$ and the value of maximum time delay $(t_{\max})$ for the best beam squint compensation as follows.
    \subsubsection{$N_t$ Selection Criterion}\label{SecIVB1} 
    Given the number of TTDs per RF chain $M$ and the $t_{\max}$ values determined by the employed TTD devices, choose $N_t$ such that 
    %\vsapce{-0.3cm}
    %\small
    \begin{equation}
         \label{eqNt}
         N_t \leq \frac{M}{2m-1} + \frac{4M}{(2m-1)}\frac{1}{\psi_{l}}f_c t_{\max}, \forall l, m,
    %\vsapce{-0.3cm}
    \end{equation}\normalsize  
    where \eqref{eqNt} is a result of substituting $N = \frac{N_t}{M}$ into \eqref{eqcondi}. The inequality in \eqref{eqNt} is rewritten as  
    %\vsapce{-0.3cm}
    %\small
   \begin{multline}
     N_t \leq \min_{l,m} \Big( \frac{M}{2m-1} + \frac{4M}{(2m-1)}\frac{1}{\psi_{l}}f_c t_{\max} \Big) \\ = \frac{M}{2M-1} + \frac{4M}{(2M-1)}\frac{1}{\max _l\psi_{l}}f_ct_{\max},\label{eq:criterion1c}
    \end{multline}
     where the $\min_{l,m}$ in \eqref{eq:criterion1c} is taken so that \eqref{eqNt} holds for all $MN_{RF}$ TTDs. The last equality in \eqref{eq:criterion1c} follows from substituting $m = M$ and $\psi_l$ by its maximum value $\max_{l}\psi_l$.
    \subsubsection{$t_{\max}$ Selection Criterion}\label{SecIVB2} 
    Equivalently, given the number of TTDs per RF chain $M$ and the number of transmit antennas $N_t$, the TTD value should be chosen to satisfy  
    %\vsapce{-0.3cm}
    %\small
    \begin{equation}
        \label{eqtmax}
         t_{\max} \geq \psi_{l}\frac{(2m-1)N_t-M}{4M}\frac{1}{f_c}, \forall l, m.
    %\vsapce{-0.3cm}
    \end{equation}\normalsize 
    The inequality in \eqref{eqtmax} can be rewritten as 
    %\vsapce{-0.3cm}
    %\small
     \begin{multline}
    t_{\max} \geq \max_{l,m}\Big(\psi_{l}\frac{(2m-1)N_t-M}{4M}\frac{1}{f_c}\Big) \\= (\max_{l} \psi_{l}) \frac{(2M-1)N_t-M}{4M} \frac{1}{f_c}.  \label{eq:criterion2c}
    \end{multline}
    For example, setting $f_c = 300$ GHz and $M=16$, the right-hande side (RHS) of \eqref{eq:criterion1c} increases linearly from $186$ to $743$ with respect to $t_{\max}$  with the slope being $6.1935 \times 10^{11}$ as $t_{\max}$ in \eqref{eq:criterion1c} grows from $300$ ps  to $1200$ ps.
    Similarly, the RHS of \eqref{eq:criterion2c} increases linearly from $412.5$ ps to $1135.8$ ps with the slope $1.6146\times 10^{-12}$ as $N_t$ in \eqref{eq:criterion2c} grows from $256$ to $720$.
     {\begin{remark}\label{rmk3.2}
          We note that since the receiver is less likely in angle $\max_l \psi_l$, one can relax $\max_l \psi_l$ in \eqref{eq:criterion1c} and \eqref{eq:criterion2c} to $\varrho\max_{l}\psi_l~(0< \varrho < 1)$ to increase the upper bound of $N_t$ to $\frac{M}{2M-1} + \frac{4M}{(2M-1)}\frac{1}{\varrho \max_{l}\psi_l}f_ct_{\max}$. 
          This allows us to bring a larger $N_t$, which readily improves the average achievable rate performance.
     \end{remark}}
\subsection{Array Gain Loss}
   Next, we characterize the array gain when \eqref{eq:criterion2c} is satisfied and array gain loss when it is not.
   First, when $t_{\max}$ meets the optimal condition in \eqref{eq:criterion2c}, we compute the array gain $g(\bff^{(l)}_{k},\psi_{k,l})$ defined in \eqref{eq7} using the optimal PS values in \eqref{eqPSa} and TTD values in \eqref{eqTTDa}, i.e., \sloppy $g(\bff^{(l)}_{k},\psi_{k,l})$ $=$ $\frac{1}{N_t}\Big| 
       \sum\limits_{m=1}^{M}  \sum\limits_{n=1}^{N}  e^{j \pi \zeta_{k} \gamma^{(l)}_{n,m}} e^{j \pi \frac{N-2n+1}{2}\psi_l} e^{-j\pi \zeta_k  \frac{(2m-1)N-1}{2}\psi_l}  \Big|$, yielding 
       \begin{equation}
           \label{eq:optimalGain}
           g(\bff^{(l)}_{k},\psi_{k,l}) =  \left| \frac{\sin(N\Delta_{k,l})}{N\sin(\Delta_{k,l})} \right|.
       \end{equation}
{
When the condition in \eqref{eq:criterion2c} is not satisfied, the array gain performance is deteriorated. 
This can be quantified in terms of the array gain loss at the $l$th RF chain $\Big|\frac{\sin(N\Delta_{k,l})}{N\sin(\Delta_{k,l})}\Big| - g(\bff_k^{(l)},\psi_{k,l})$ as shown in the following remark.
\begin{remark}\label{lmAGloss}
Suppose that the first $M' \leq M$ TTDs are given by \eqref{eqTTDa} and the last $(M-M')$ are given by \eqref{eqTTDb}, i.e., $\psi_l\frac{(2M'-1)N_t-M}{4M}\frac{1}{f_c} \leq t_{\max} \leq \psi_l\frac{(2(M'+1)-1)N_t-M}{4M}\frac{1}{f_c}$.
Then, the following holds,
\begin{equation}\label{eqAGlossBound}
    \d4\scalemath{0.85}{0 \leq \Big|\frac{\sin(N\Delta_{k,l})}{N\sin(\Delta_{k,l})}\Big| - g(\bff_k^{(l)},\psi_{k,l}) \leq \frac{M-M'}{M}\Big(\Big|\frac{\sin(N\Delta_{k,l})}{N\sin(\Delta_{k,l})}\Big| + 1\Big).}
\end{equation}
\end{remark}
 The proof of \eqref{eqAGlossBound} is relegated to Appendix~\ref{AppendixE}. 
 The bound in \eqref{eqAGlossBound} implies that the array gain loss becomes zero when $M'= M$, which occurs when $t_{\max}$ is sufficiently large to meet the optimal condition in \eqref{eq:criterion2c} and the TTD values $\{t^{(l) \star}_m\}_{m=1}^{M}$ are given by \eqref{eqTTDa}.}
{The array gain loss upper bound decreases by the factor of $\frac{1}{M} \Big( \Big| \frac{\sin(N\Delta_{k,l})}{N\sin(\Delta_{k,l})}\Big| + 1\Big)$ when $M'$ increases to $(M'+1)$ while satisfying $\psi_l\frac{(2(M'+1)-1)N_t-M}{4M} \frac{1}{f_c} \leq t_{\max} \leq \psi_l \frac{(2(M'+2)-1)N_t-M}{4M} \frac{1}{f_c}$.} 
{
The TTD-based hybrid precoding architecture in Fig.~\ref{fig2} utilizes $N_{RF}$ RF chains, $MN_{RF}$ TTDs, and $N_{RF}N_t$ PSs, which could increase hardware complexity substantially.
%%
%This challenge primarily arises from the large numbers of TTDs $(M)$, antennas $(N_t)$, and analog combiners $(N_t)$ needed to combine signals from all RF chains.
%%
%%
Reducing the hardware complexity of the architecture in Fig.~\ref{fig2} while maintaining the beam squint compensation performance is an important research problem and is subject to further study.}
In this connection, the system parameter selection criteria in \eqref{eq:criterion1c} and \eqref{eq:criterion2c} give an idea of the values of $N_t$ and $t_{\max}$ for the best beam squint compensation given a fixed number of TTDs per RF chain $M$.
In practice, it is also critical to install an appropriate value of $M$ to reduce the power consumption of the analog precoder, which is addressed in the next subsection.
\subsection{$M$ Selectrion Criterion}
    \label{secIVC}
    In this subsection, we further exploit the closed-form solution of the TTD and PS precoders in Theorem~\ref{theorem1} to characterize the minimum number of TTDs per RF chain $M$ given a predefined array gain performance while assuming that the constraint in \eqref{eq:criterion2c} is always satisfied and the number of transmit antennas $N_t$ is fixed.     
    % %%
    To this end, we formulate a problem that minimizes $M$ to guarantee a given array gain performance:  
    %\vsapce{-0.3cm}
    \begin{equation}
    \label{eqPowerGain}
        %\beq
        \min_{M} M 
        \text{ subject to }  N_t = MN,~  g(\bff^{(l)}_k,\psi_{k,l}) \geq g_0,\forall k,l, 
    %\vsapce{-0.3cm}
    \end{equation}  
    where the second constraint in \eqref{eqPowerGain} describes the array gain guarantee with the threshold $0 < g_0 < 1$.
    We note that the total power consumption of the analog precoder can be modeled as $P_{total} = N_{RF}MP_{TTD} + N_{RF}N_tP_{PS}$, where $P_{TTD}$ and $P_{PS}$ are the 
    power consumption of a TTD and a PS, respectively. Hence, given the fixed values of $N_t$, $P_{TTD}$, and $P_{PS}$, the objective of minimizing the power consumption (i.e., $\min_{M}P_{total}$) is equivalent to the objective in \eqref{eqPowerGain}. 

    Incorporating $g(\bff^{(l)}_{k},\psi_{k,l})$ in \eqref{eq:optimalGain} into the second constraint in \eqref{eqPowerGain} and substituting $\frac{N_t}{M}$ for $N$, the problem is converted to:
%\vsapce{-0.2cm}
\begin{equation}
\label{eqMt}
\min_{M}   M
 \text{ subject to }   \bigg|\frac{\sin\left(\frac{N_t}{M}\Delta_{k,l}\right)}{\frac{N_t}{M}\sin(\Delta_{k,l})} \bigg| \geq g_0,\forall k,l. 
%\vsapce{-0.2cm}
\end{equation} 
The problem in \eqref{eqMt} can be viewed as finding the minimum element in the feasible set \small$\cM$ $=$ $\left( \bigcap_{k=1,l=1}^{K,L} \cM_{k,l} \right) \cap \cN_t$\normalsize, where \small$\cM_{k,l} = \bigg\{M \in  \R \colon \Big|\frac{\sin(\frac{N_t}{M}\Delta_{k,l})}{\frac{N_t}{M}\sin(\Delta_{k,l})} \Big| \geq g_0 \bigg\}$\normalsize, for $1 \leq k \leq K$ and $1 \leq l \leq L$, and \small$\cN_t = \left\{m \in \N \colon m | N_t\right\}$\normalsize, where $m| N_t$ denotes that $m$ is a divisor of $N_t$.
Thus, solving the problem \eqref{eqMt} using the greedy search method requires the construction of $KL+1$ sets $\{\cM_{k,l}\}$ and $\cN_t$, which is demanding as the number of OFDM subcarriers grows. 
To cope with these difficulties, we propose to 
approximate the constraint in $\cM_{k,l}$, i.e., the constraint in \eqref{eqMt} as $|q(\Delta_{k,l})| \geq g_0, \forall k,l,$ where $q(x)=1+\frac{1}{6}\left(1-\frac{N^2_t}{M^2}\right)x^2$, which is obtained by the Taylor expansion $\frac{\sin\left(\frac{N_t}{M}x\right)} {\frac{N_t}{M}\sin(x)} = q(x) + \cO(x^4)$ derived at $x =0$. Noting that $|q(\Delta_{k,l})| \geq q(\Delta_{k,l})$, the problem in \eqref{eqMt} is relaxed to  
%\vsapce{-0.5cm}
\begin{equation}
\label{eqMopt1}
M^{\star} =\argmin_{M}  M \text{ subject to }  M \in \cN_t, ~ q(\Delta_{k,l}) \geq g_0,\forall k,l. 
%\vsapce{-0.5cm}
\end{equation}

About the optimal solution to \eqref{eqMopt1}, denoted as $M^{\star}$, we have the following theorem.  
 \begin{figure*}[htb]
        \centering
        \subfloat[]{%
        \includegraphics[scale=.45,trim=0.5cm 0.0cm 0.0cm 0.0cm]{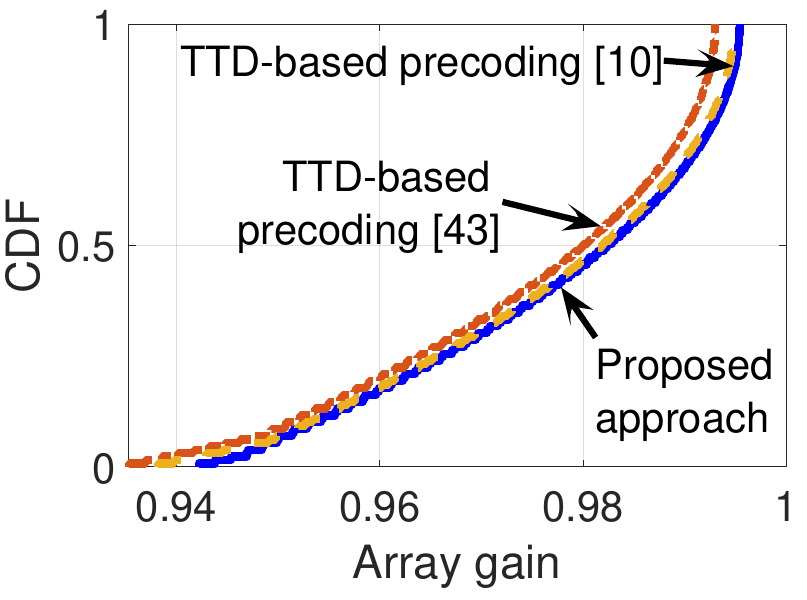}%
        \label{figV1a}%
        }%
%        \hfill
        %\centering
        \subfloat[]{%
        \includegraphics[scale=.45,trim=0.0cm 0.0cm 0.0cm 0cm]{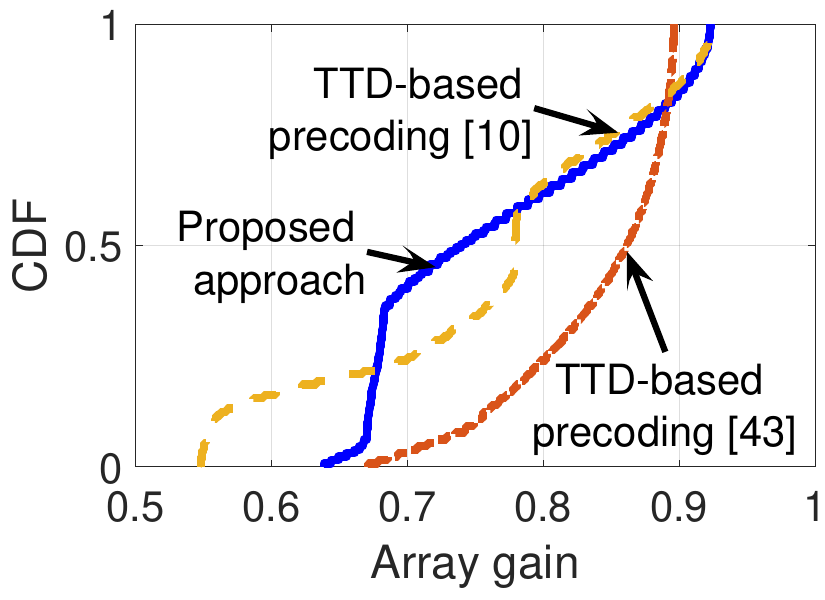}%
        \label{figV1b}%
        }%
%     \hfill
         \subfloat[]{%
        \includegraphics[scale=.45,trim=0.0cm 0.0cm 0.0cm 0cm]{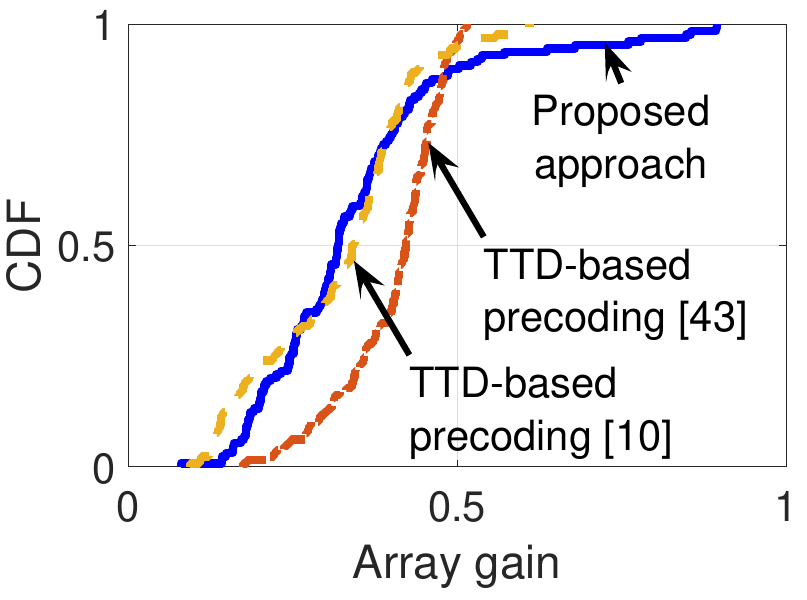}%
        \label{figV1c}%
        }%
        \caption{\!CDF of the array gain with 
        different number of transmit antennas ($N_t$) values:\! (a) $N_t \!=\! 128$, (b) $N_t \!=\! 256$, and c) $N_t \!=\! 512$.}
    \label{figV1}
 \end{figure*}
 \begin{figure*}[htb]
        \centering
        \subfloat[]{%
        \includegraphics[scale=.45,trim=0.5cm 0.0cm 0.0cm 0.0cm]{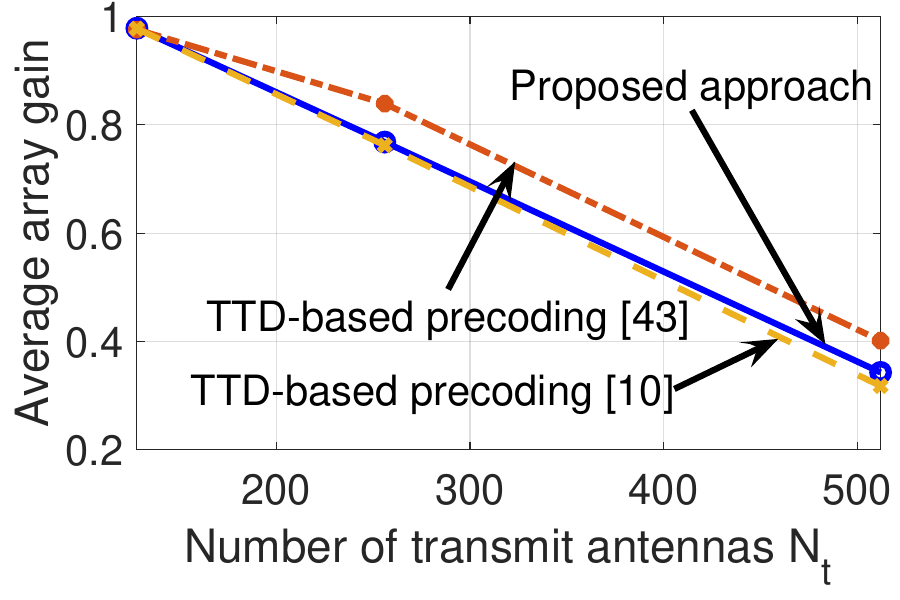}%
        \label{figV2a}%
        }%
        \hspace{1cm}
        \subfloat[]{%
        \includegraphics[scale=.45,trim=0.0cm 0.0cm 0.0cm 0cm]{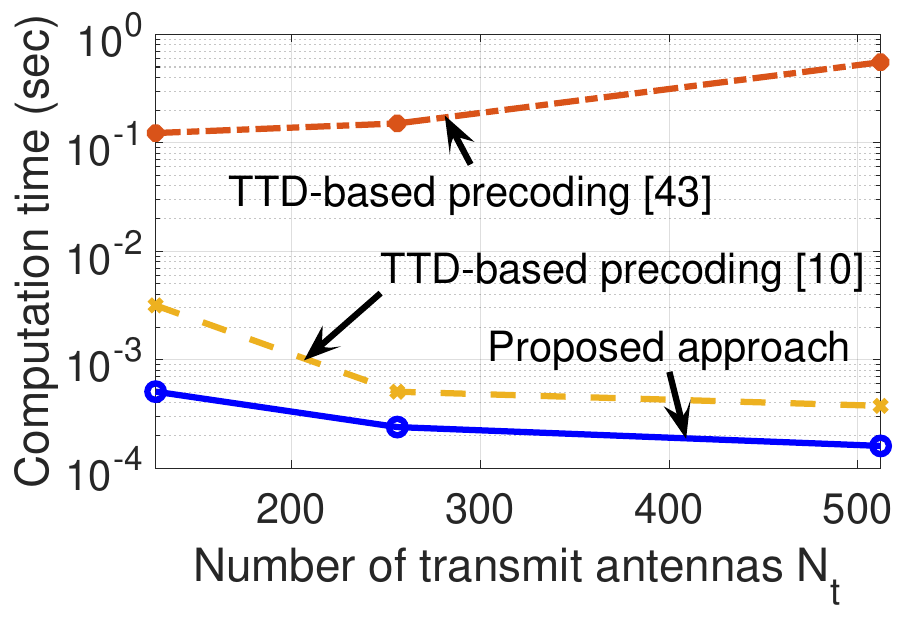}%
        \label{figV2b}%
        }%
        \caption{(a) Average array gain vs. $N_t$ and (b) computation time to design PS and TTD values when $N_t \in \{128,256,512\}$.}
    \label{figV2}
 \end{figure*}
\begin{theorem}
    \label{theorem2}
    Given the fixed  $N_t$, $B$, $f_c$, and $K$, the  ${M}^{\star}$ of \eqref{eqMopt1} is      ${M}^{\star} =  \left\lceil~\sqrt{\frac{N^2_t}{1+\Omega(g_0,B)}}~\right\rceil_{N_t}$,  
    where $\Omega(g_0,B) = \frac{6(1-g_0)}{\left(\frac{\pi}{2} \frac{B}{f_c} \frac{K-1}{2K}\right)^2 \max_l \psi^2_l }$ and $\lceil x \rceil_{N_t}$ denotes the smallest integer greater than or equal to $x$ that is a divisor of $N_t$. 
%\vsapce{-0.3cm}
\end{theorem}
\begin{proof}
    See Appendix~\ref{AppendixD}. 
%\vsapce{-0.3cm}
\end{proof}
\begin{remark}
\label{rmk3}  
To understand the relationship between the $M^{\star}$ in Theorem~\ref{theorem2} and the system parameters, we first relax its integer constraint, leading to \small$M^{\star} = \sqrt{\frac{N^2_t}{1+\Omega(g_0,B)}}$\normalsize. 
In the regime of a large number of OFDM subcarriers ($K \gg 1$), we have $\frac{K-1}{2K} \approx \frac{1}{2}$, leading to   
\small$M^{\star} \approx \sqrt{\frac{N^2_t (\frac{\pi}{4}\frac{B}{f_c})^2\max_l\psi^2_l}{6(1-g_0)}} = \left(\frac{\pi N_t}{4f_c}\sqrt{\frac{\max_l\psi^2_l}{6(1-g_0)}} \right) B$\normalsize.   
Given a fixed $N_t$ and $t_{\max}$, it reveals that the minimum required number of TTDs that ensures a predefined array gain performance grows linearly with the bandwidth $B$. 
{For instance, setting $N_t = 720$ and  $t_{\max} = 1200$ ps meets the optimal condition in (31), the required number of TTDs $M^{\star}$ increases approximately linearly from 20 to 180 as $B$ grows from $10$ GHz to $80$ GHz with the slope $2.4335 \times 10^{-9}$ for $g_0 = 0.9$.} 
\end{remark}
    The proposed approach embodies a fundamental engineering design principle: maximizing system performance within existing hardware constraints. 
    Striking this balance is crucial as it not only improves resource utilization but also enhances the overall performance and feasibility of TTD-based hybrid precoding systems. 
    The proposed approach aligns with the resource-conscious design principle.
\section{Simulation Results}
\label{secV}
{In this section, we demonstrate the benefits of the proposed joint TTD and PS precoding by comparing it with the state-of-the-art closed-form solution in \cite{dai2021} and iterative design with a least square solution in [43, Algorithm 1] in terms of array gain and computation time required to design TTD and PS values.
We set the number of iterations for the benchmark [43, Algorithm 1] as $N_{iter} = 10$ to alternatingly design PS and TTD values. 
The numerical results are conducted on an Intel(R) Core(TM) i7-9850H CPU using MATLAB 2023a.}
Throughout the simulation, the system parameters are set as follows unless otherwise stated: 
the central carrier frequency $f_c \= 300$ GHz, bandwidth $B \= 30$ GHz, the number of OFDM subcarriers $K \= 129$,  number of transmit antennas $N_t \= 256 $,  number of receive antennas $N_r \= 4$,  number of RF chain $N_{RF} \= 4$, and number of data stream $N_{s} \= 4$. 
The number of TTDs per RF chain $M \= 16$, the number of PSs per TTD $N \=16$, and the maximum time delay  {$t_{\max} = 320$ ps}. 
The AoDs and AoAs are uniformly drawn in $[\frac{-\pi}{2},\frac{\pi}{2}]$.  
\subsection{Array Gain Performance and Computation Time}
\label{secVA}
 \begin{figure*}[htb]
        \centering
        \subfloat[]{%
        \includegraphics[scale=.42,trim=0.5cm 0.0cm 0.0cm 0.0cm]{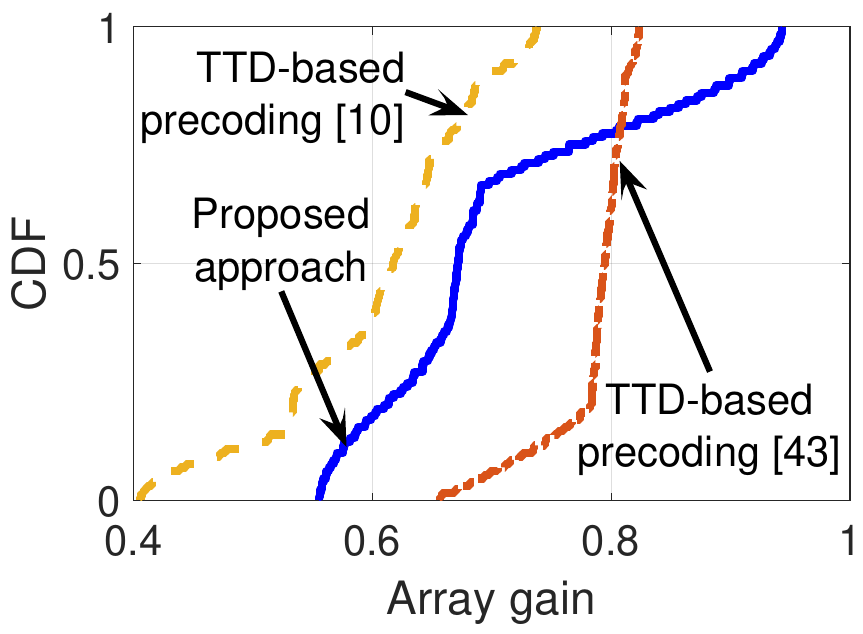}%
        \label{figV3a}%
        }%
%\hspace{1mm}
      %  \hfill
        \subfloat[]{%
        \includegraphics[scale=.44,trim=0.cm 0cm 0.0cm 0cm]{figV1b.pdf}%
        \label{figV3b}%
        }%
%      \hfill
%\hspace{1mm}
         \subfloat[]{%
        \includegraphics[scale=.42,trim=0.0cm 0.0cm 0.0cm 0cm]{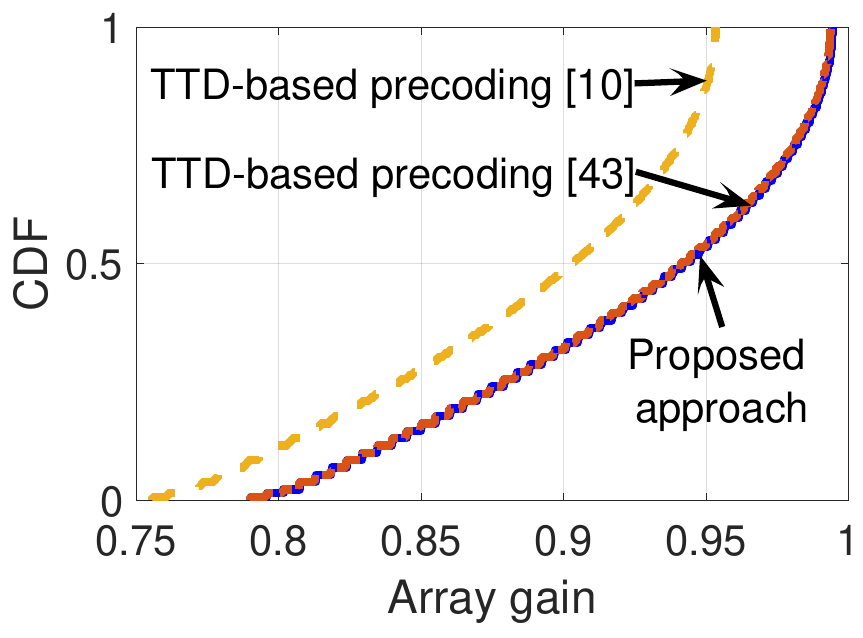}%
        \label{figV3c}%
        }%
       %%\vspace{-0.5cm}
        \caption{\!CDF of array gain with different maximum time delay ($t_{\max}$) values: (a) {$t_{\max} \!=\! 280 $ ps}, (b) $t_{\max} \!=\! 320$ ps, and (c) {$t_{\max} \!=\! 380$ ps}.}
    \label{figV3}
    %\vspace{-0.5cm}
    \end{figure*}
 \begin{figure*}[htb]
        \centering
        \subfloat[]{%
        \includegraphics[scale=.45,trim=0.5cm 0.0cm 0.0cm 0.0cm]{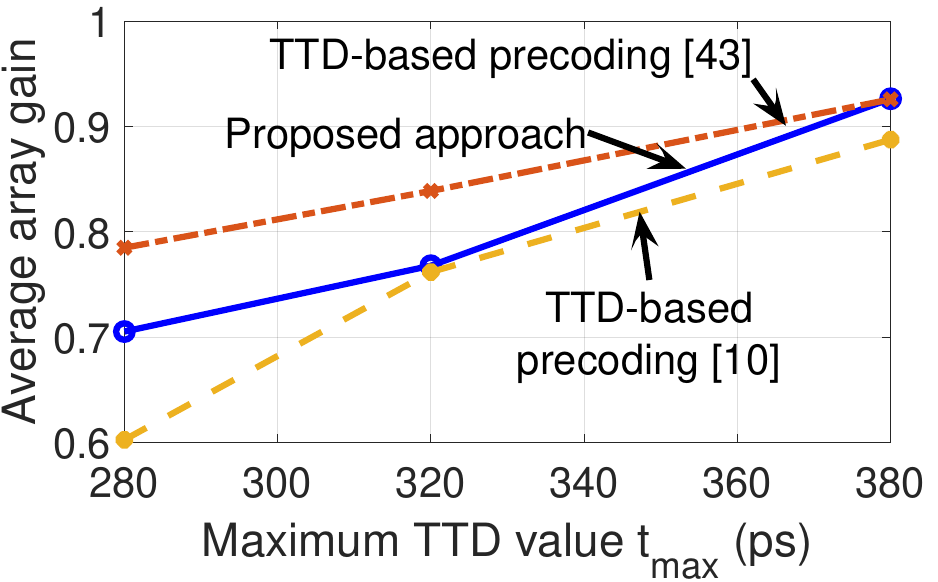}%
        \label{figV4a}%
        }%
        \hspace{1cm}
        \subfloat[]{%
        \includegraphics[scale=.45,trim=0.0cm 0.0cm 0.0cm 0cm]{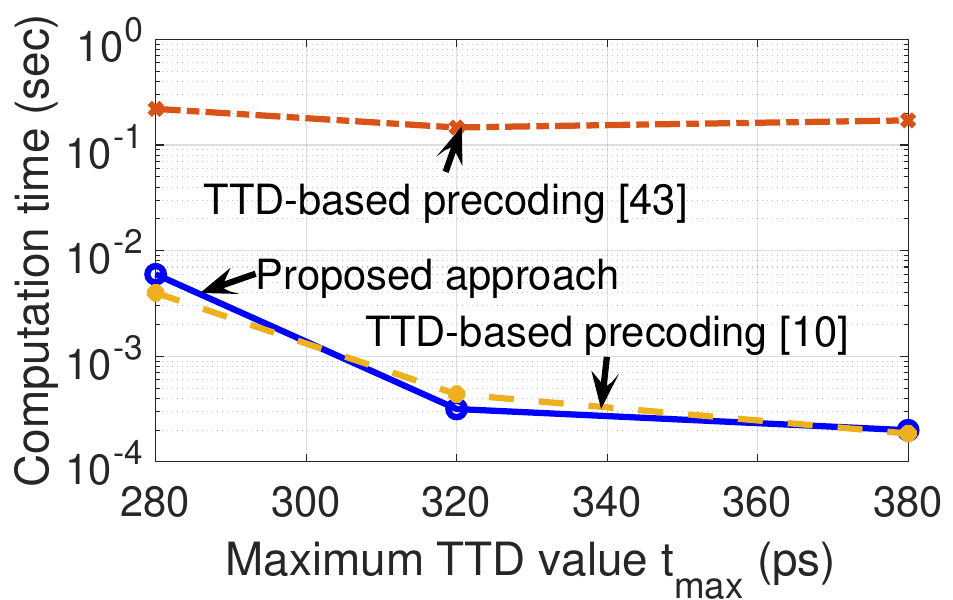}%
        \label{figV4b}%
        }%
        \caption{(a) Average array gain vs. $t_{\max}$ and (b) computation time to design PS and TTD values when $t_{\max} \in \{280,320,380\}$ ps and TTD step size $2$ ps.}
    \label{figV4}
 \end{figure*}

{
Although the derivations in the paper assume continuous PS and TTD values, the following numerical evaluations adopt finite-resolution PS and TTD values to be realistic.}
In particular, we specify the set of feasible PS values as $\mathcal{P}$ with 4-bit quantization, represented by $\mathcal{P} = \{0, \frac{2\pi}{16}, \dots, \frac{15}{16}2\pi\}$.
Similarly, the set of feasible TTD values $\mathcal{T}$ is defined with a $2$-ps-step size, such that $\mathcal{T} = \{0, 2, \dots, 318, 320\}$ ps, accommodating the practical limitations of TTD devices.
Quantization is applied directly to the PS and TTD values derived in Theorem 1, aligning them with the nearest values in $\mathcal{P}$ and $\mathcal{T}$, respectively, to obtain a solution of discrete PS and TTD values. 
To ensure a fair evaluation, this quantization method is also applied to the closed-form solution in benchmark \cite{dai2021} and the iterative least squares solution detailed in \cite[Algorithm 1]{Ratnam2022}.\looseness=-1
 We measure the empirical cumulative distribution function (CDF) of the array gain  $G_{\psi}(x)$ and {the average array gain $\frac{1}{K}\sumK g(\bff^{(l)}_k,\psi_k)$ at the central spatial direction $\psi$. We also measure the corresponding computation time to design PS and TTD values.}
Herein, the CDF of the array gain is computed as  $G_{\psi}(x) = \frac{1}{K}\sumK \mathds{1}_g(k,x)$, where $x$ represents the array gain $0 \leq x \leq 1$ in (4), the indicator function  $\mathds{1}_g(k,x)$ takes values $1$ if $g(\bff^{(l)}_k,\psi_{k}) \leq x$ and $0$ otherwise, and $g(\bff^{(l)}_k,\psi_{k})$ is computed as in (4) with the spatial direction {$\psi = 0.9$}. 
 
Fig.~\ref{figV1} displays the empirical CDF curves of the array gain when the number of transmit antennas $N_t$ takes the values from $\{128,256,512\}$. 
    {When $N_t = 128$, the proposed approach exhibits a superior array gain performance compared to the benchmarks \cite{dai2021,Ratnam2022}.}
    When $N_t = 256$ (Fig.~\ref{figV1b}) and $N_t = 512$ (Fig.~\ref{figV1c}), the proposed approach and benchmarks \cite{dai2021}, \cite{Ratnam2022} suffer from the array gain loss because the degree of beam squint increases as $N_t$ grows (i.e., Proposition~1).
    Nevertheless, the proposed approach provides an enhanced beam squint compensation capability compared to the benchmark \cite{dai2021}. 
    This is due to the fact that as $N_t$ grows, some of the time delay values of the benchmark \cite{dai2021} become larger than $t_{\max} = 320$ ps, in which case the time delay value needs to be quantized to the $t_{\max}$. 
    {Yet, when $N_t = 256$ and $N_t = 512$, the benchmark \cite{Ratnam2022} outperforms the proposed approach and benchmark \cite{dai2021}.
    Notably, in Fig.~\ref{figV1c}, approximately 65\% of OFDM subcarriers of the benchmark \cite{Ratnam2022} achieve array gains $\geq 0.4$, while only 25\% of OFDM subcarriers of the proposed approach achieve similar performance.} %%
    Fig.~\ref{figV2a} displays the corresponding average array gain performances of the array gain CDF curves displayed in Fig.~\ref{figV1}.
     {Observations from Fig.~\ref{figV2a} reveal the proposed approach achieves slightly higher average array gain than the benchmarks \cite{dai2021, Ratnam2022} when $N_t = 128$. 
    When beam squint becomes more severe ($N_t = 256$ and $N_t = 512$), the benchmark \cite{Ratnam2022} achieves a moderately higher average array gain than our proposed approach and benchmark \cite{dai2021}.} 
    This is also consistent with the CDF curves in Fig.~\ref{figV1c}.
Fig.~\ref{figV2b} illustrates the computation time required by our proposed approach and the benchmarks \cite{dai2021}, \cite{Ratnam2022} to achieve the average array gain performances shown in Fig.~\ref{figV2a}. 
The comparison reveals that the proposed approach and benchmark \cite{dai2021} maintain a rather consistent time complexity trend as $N_t$ grows, underscoring the efficiency of one-shot low-complexity designs. 
Conversely, the computation time of the benchmark \cite{Ratnam2022} increases substantially with $N_t$, attributed to its iterative design process which necessitates matrix inversion at each step—a computationally intensive operation.
For instance, at $N_t = 256$, the proposed approach records a swift computation time of {$2.4\times 10^{-4}$} sec, markedly faster compared to the {$0.1508$} sec required by the benchmark \cite{Ratnam2022}. 
This highlights the computational advantage of our one-shot design over iterative methods that involve complex calculations for each round of iteration.
Fig.~\ref{figV3} demonstrates the empirical CDF curves of the array gain when the maximum time delay  $t_{\max}$ takes values in $\{280, 320, 380\}$ ps while fixing $N_t = 256$ and setting the TTD step size to $2$ ps.  
It is observed in Fig.~\ref{figV3} that the proposed approach outperforms the benchmark \cite{dai2021}.  
For example, it can be observed from Fig.~\ref{figV3a} that 33\% of OFDM subcarriers of the proposed approach achieve $\geq 0.7$ array gain compared to that $10\%$ of OFDM subcarriers of the benchmark \cite{dai2021} attain the similar performance. 
{Similar to the trend in Figs.~\ref{figV1b}-\ref{figV1c}, it can be observed from Figs.~\ref{figV3a}-\ref{figV3b} that the benchmark \cite{Ratnam2022} outperforms the proposed approach when the maximum TTD value is not sufficient to compensate for the beam squint.} 
The proposed approach and the benchmark \cite{Ratnam2022} converge to the same array gain performance when $t_{\max} = 380$ (ps).
Fig.~\ref{figV4a} demonstrates the average array gain performance of the corresponding CDF curves shown in Fig.~\ref{figV3}.
{Seen from Fig.~\ref{figV4a}, the proposed approach exhibits a superior average array gain performance compared to the benchmark \cite{dai2021}.}
When $t_{\max} = 380$ ps, the proposed approach and the benchmark {\cite{Ratnam2022}} converge to a similar average array gain performance. 
These observations are consistent with the trends in Fig.~\ref{figV3}.\looseness=-1
{Fig.~\ref{figV4b} illustrates the computation time needed by the proposed approach and benchmarks \cite{dai2021}, \cite{Ratnam2022} to achieve the average array gain performance depicted in Fig.~\ref{figV4a}. 
In alignment with the trends observed in Fig.~\ref{figV2b}, Fig.~\ref{figV4b} corroborates that the proposed approach and the benchmark \cite{dai2021} demand significantly less computation time compared to the benchmark {\cite{Ratnam2022}}. 
Notably, at $t_{\max} = 380$ ps, the benchmark {\cite{Ratnam2022}} necessitates a computation time of {0.146} sec while the proposed approach achieves comparable performance with an impressively reduced computation time of only {$3.16\times 10^{-4}$} sec.}
\subsection{Array Gain Performance Guarantee of 
$\widetilde{M}^{\star}$ in Theorem~\ref{theorem2}}    
\label{secVC}
\begin{figure}[htb]
    \centering
      \includegraphics[width = 0.4\textwidth]{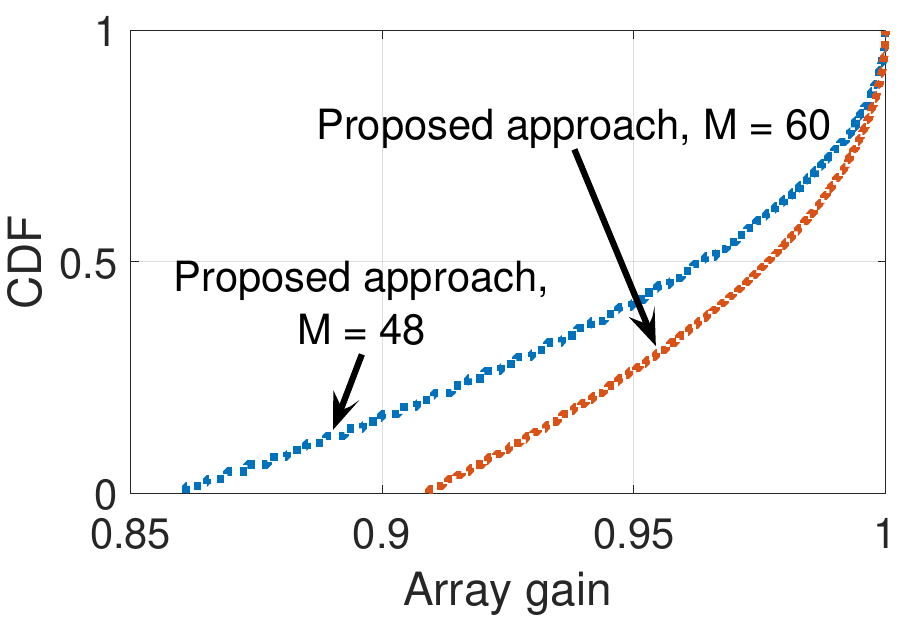}%
        \caption{CDF of array gain of the proposed approach with different number of TTDs per RF chain ($M$) values.}
       \label{figV7}
\end{figure}

In this subsection, we demonstrate the array gain performance guarantee provided by the minimum required number of TTDs per RF chain $\widetilde{M}^{\star}$ in Theorem 2.
Similar to Section~V-A, we evaluate the empirical CDF of the array gain $G_{\psi}(x)$ at the spatial direction $\psi = 0.8$ except that we assume $N_t = 720$ and $t_{\max} = 1000$ ps in this simulation {with 8-bit quantization for the PS values and 2-ps-step size for the TTD values.} 
The array gain threshold is set to $g_0 = 0.9$ in Theorem~2. Incorporating $g_0$ and the system parameters into Theorem~2 yields   $\widetilde{M}^{\star} = \left\lceil~ \sqrt{\frac{720^2}{1+\Omega(0.9,30)}} ~\right\rceil_{720} = 60$.

Fig.~\ref{figV7} demonstrates the CDF of the array gain when $M = \widetilde{M}^{\star} = 60$ and $M = 48$, where $M =48$ is chosen to be the largest divisor of $N_t$ that is smaller than $ \widetilde{M}^{\star} = 60$.
We note that the condition in (31) is satisfied for both values of $M$ when $t_{\max}=1000$ ps.
As shown in Fig.~\ref{figV7}, for the optimized $ \widetilde{M}^{\star} = 60$, every OFDM subcarrier satisfies the array gain $\geq 0.9$. 
However, when $M=48$, there are $18\%$ of the OFDM subcarriers that have the array gain $< 0.9$. 
Hence, the curves in Fig.~\ref{figV7} verify that the $\widetilde{M}^{\star} = 60$ is the minimum number of TTDs per RF chain to provide the minimum array gain $g_0 = 0.9$.
\section{Conclusion}
\label{secVI}
    We presented a new framework to the problem of compensating the beam squint effect arising in wideband sub-THz hybrid massive MIMO systems. 
    We determined the ideal analog precoder that fully compensates for the beam squint. 
    A novel TTD-based hybrid precoding approach was proposed  by jointly optimizing the TTD and PS precoders under the 
    per-TTD time delay constraints. 
    The joint optimization problem was formulated in the context of minimizing the distance between the ideal analog precoder and the product of the PS and TTD precoders.
    By transforming the original problem into the phase domain, the original problem was converted to an {approximated problem}, which allowed us to find a closed-form solution. 
   Based on the closed-form expression of our solution, we presented the selection criteria for the required number of transmit antennas and the value of maximum time delay. 
   Exploiting the proposed joint TTD and PS precoder 
   optimization approach, we determined the minimum number of TTDs required to achieve an array gain performance guarantee while minimizing the analog precoder power consumption.   
 {Simulations corroborated the advantages of our joint TTD and PS optimization approach by demonstrating guaranteed array gain performance and reduced computational time.}
\section*{{Acknowledgment}} 
\addcontentsline{toc}{section}{Acknowledgment}
{The authors are grateful to Reviewer 3 for insightful technical comments and the shared source codes for evaluating the array gain performance of \cite{Ratnam2022}. These contributions have significantly enhanced the quality of this work.}
 \appendices
 \section{Proof of lemma~\ref{lm1}}
  \label{AppendixB}
     Without loss of generality, we assume  $0 < |x - y| < \pi$. 
     Then, ${\argmin}_{{0 <|x-y| < \pi }}|e^{jx}-e^{jy}| = {\argmin}_{0 <|x-y|<\pi}\Big|\sin\Big(\frac{x-y}{2}\Big)\Big|={\argmin}_{0 <|x-y| <\pi}\sin\Big(\Big|\frac{x-y}{2}\Big|\Big)= \argmin_{0 <|x-y| <\pi} |x-y|$\normalsize, 
    where the last step follows from the fact that $\sin(z)$  is an increasing function of $0 < z < \frac{\pi}{2}$. This completes the proof.\looseness=-1 
\section{Proof of Theorem~\ref{theorem1}}
    \label{AppendixTheorem1}
    To prove {Theorem}~\ref{theorem1}, we start by formulating the Lagrangian of  \eqref{eq:lcqp}, which is given by 
    \begin{multline}
        \label{eq:Lagrangian}
        \cL(\ba^{(l)}_m,\lambda_1, \lambda_2) =  \ba^{(l)T}_m\bC\ba^{(l)}_m -2\bd^{(l) T}_m\ba^{(l)}_m +         \\ \lambda_1(\be_{N+1}^T\ba^{(l)}_m -  \vartheta_{\max}) +  \lambda_2(-\be_{N+1}^T\ba^{(l)}_m),
    \end{multline}
     where $\lambda_1 \geq 0$ and $\lambda_2 \geq 0$ are the Largrangian multipliers.
    After incorporating the first order necessary condition for $\ba_m^{(l)}$ in \eqref{eq:Lagrangian},  the Karush–Kuhn–Tucker (KKT) conditions of \eqref{eq:lcqp} are given by: 
    \begin{equation}
         \left\{ \begin{array}{ll}
        \label{KKT}
        2\bC \ba^{(l)}_m -2\bd^{(l)}_m + \lambda_1 \be_{N+1} - \lambda_2\be_{N+1} & = 0,  \\
         \lambda_1(\be_{N+1}^T\ba^{(l)}_m -  \vartheta_{\max}) &= 0, \\ 
         \lambda_2(-\be_{N+1}^T\ba^{(l)}_m) &= 0,\\
        \lambda_1 \geq 0, \lambda_2 \geq 0. 
        \end{array} \right.
    \end{equation}\normalsize   
    When $\lambda_1 = \lambda_2 = 0$, \eqref{KKT} yields $\ba^{(l)}_{m} = \bC^{-1}\bd^{(l)}_m$ for $0 \leq \be_{N+1}^T\bC^{-1}\bd^{(l)}_m \leq \vartheta_{\max}$ because $0 \leq \be_{N+1}^T\ba^{(l)}_m \leq \vartheta_{\max}$. 
    When $\lambda_1 > 0$ and $\lambda_2 = 0$, \eqref{KKT} gives $\ba^{(l)}_m = \bC^{-1}(\bd^{(l)}_{m} - \frac{1}{2}\lambda_1 \be_{N+1})$ for $ \be_{N+1}^T\bC^{-1}\bd^{(l)}_{m} > \vartheta_{\max}$ because $\be^T_{N+1}\ba^{(l)}_m = \vartheta_{\max}$ and $\be^T_{N+1}\ba^{(l)}_m = \be_{N+1}^T\bC^{-1}\bd^{(l)}_{m} - \frac{1}{2}\lambda_1\frac{1}{\eta}$, where the last equality follows from \eqref{eqCe}.  
    When $\lambda_1 = 0$ and $\lambda_2 >0$, we obtain $\ba^{(l)}_m = \bC^{-1}(\bd^{(l)}_m + \frac{1}{2}\lambda_2\be_{N+1})$ for $\be_{N+1}^{T}\bC^{-1}\bd^{(l)}_m < 0$ because $\be^T_{N+1}\ba^{(l)}_m = 0$ and $\be^T_{N+1}\ba^{(l)}_m = \be^T_{N+1}\bC^{-1}\bd^{(l)}_m + \frac{1}{2}\lambda_2\frac{1}{\eta}$, where the last equality follows from \eqref{eqCe}. 
    However, by \eqref{eqCde} we have $\be_{N+1}^{T}\bC^{-1}\bd^{(l)}_m = \frac{(2m-1)N-1}{2}\psi_{l} \geq 0$, where the last inequality is due to the sign invariance property in \eqref{eqSIP}. 
    Therefore, the last case leads to a contradiction. 
    In summary, solving \eqref{KKT} gives
    %%
    %%
   %%\vspace{-0.1cm}
    %\small
     \begin{subnumcases}
     {\label{eqCFKKT} {\scalemath{0.9}{\ba^{(l)}_{m}}^{\star}} = }
    \scalemath{0.9}{\bC^{-1}\bd^{(l)}_m}, ~$\text{\normalsize if } \scalemath{0.9}{0 \leq \frac{(2m-1)N-1}{2}\psi_{l} \leq \vartheta_{\max}}$,  \label{eqCFKKTa} \\
    \scalemath{0.9}{\bC^{-1}\left(\bd^{(l)}_{m} -\frac{1}{2} \lambda_1 \be_{N+1}\right)},  ~$\text{\normalsize otherwise,}$ \label{eqCFKKTb}
    \end{subnumcases}
    \normalsize
    where $\lambda_1$ $=$ $2\eta\left(\frac{(2m-1)N-1}{2}\psi_l-\vartheta_{\max}\right)$. 
    
    To further simplify the closed-form solution in \eqref{eqCFKKT}, we deduce ${\ba^{(l)}_{m}}^{\star}$ in \eqref{eqCFKKTa} by simplifying $\bC^{-1}\bd^{(l)}_m$. Using \eqref{eqdlm}  gives
    %\begin{dmath}
    %\vspace{-0.3cm}
    %\small
    \begin{equation}
    \label{eqCdlm}
    \d4\scalemath{0.9}{       \bC^{-1}\bd^{(l)}_m 
      = \begin{bmatrix} \frac{1}{K}\sumK \bb^{(l)}_{k,m} +\frac{1}{\eta}\frac{1}{K}\sumK(1-\zeta_k) \mathbf{1}_N\mathbf{1}^T_N\bb^{(l)}_{k,m} \\ \frac{1}{\eta}\frac{1}{K}\sumK(1-\zeta_k) \mathbf{1}^T_N\bb^{(l)}_{k,m} \end{bmatrix}\!\!.}\!\!
      %\vspace{-0.2cm}
    \end{equation}
    \normalsize
    %\end{dmath}
    %Noting that $\bb^{(l)}_{k,m} \= [-\zeta_k\gamma^{(l)}_{1,m}~\dots~-\zeta_k\gamma^{(l)}_{N,m}]^T$, 
    Using the definition of $\bb^{(l)}_{k,m}$ $=$ $[-\zeta_k\gamma^{(l)}_{1,m}~\dots~-\zeta_k\gamma^{(l)}_{N,m}]^T$, the $n$th entry of $\widetilde{\bd}^{(l)}_m \triangleq \bC^{-1}\bd^{(l)}_m$ in \eqref{eqCdlm}, where $1\leq n \leq N$, is given by 
    %\vspace{-0.2cm}
    %\small 
    \begin{equation}
    \begin{aligned}
             \widetilde{\bd}^{(l)}_m(n,1) &= \frac{1}{K} \gamma^{(l)}_{n,m}\sumK-\zeta_k + \frac{1}{\eta K} \sumN \gamma^{(l)}_{n,m} \sumK (\zeta^2_k-\zeta_k) \\&= \frac{N-2n+1}{2}\psi_l,  \label{eqCdlm2c}
    \end{aligned}
    \end{equation}
    where the last equality in  \eqref{eqCdlm2c} is due to \eqref{eqZeta} and the fact that $\gamma^{(l)}_{n,m} = ((m-1)N+n-1)\psi_l$. The $(N+1)$th entry of $\widetilde{\bd}^{(l)}_m$ in \eqref{eqCdlm} is 
    %\small
    \begin{equation}
    \label{eqCdlm3}
        \widetilde{\bd}^{(l)}_m(N+1,1) =  \frac{1}{N}\sumN \gamma^{(l)}_{n,m} = \frac{(2m-1)N-1}{2}\psi_l.   
    \end{equation}\normalsize 
    Therefore, substituting \eqref{eqCdlm2c} and \eqref{eqCdlm3} into \eqref{eqCFKKTa} leads to \eqref{eqPSa} and \eqref{eqTTDa}, respectively.   
    Next, \eqref{eqCFKKTb} can be rewritten as      
    \begin{equation}
    \label{eqcond2}
\scalemath{0.9}{    \bC^{-1}(\bd^{(l)}_m-\frac{1}{2}\lambda_1\be_{N+1}) \!=\! \widetilde{\bd}^{(l)}_m \!-\!  \left(\!\frac{(2m-1)N-1}{2}\psi_l\!-\!\vartheta_{\max}\!\right)\!\mathbf{1}_{N+1}. }
    \end{equation}
    \normalsize
    Incorporating \eqref{eqCdlm2c}  and \eqref{eqCdlm3} into \eqref{eqcond2} leads to \eqref{eqPSb} and \eqref{eqTTDb}, respectively,  completing the proof.\looseness=-1  
{\section{Proof of Remark~\ref{lmAGloss}}\label{AppendixE}
From the assumption in Remark~\ref{lmAGloss}, we have $t^{(l)\star}_{m} = \frac{(2m-1)N-1}{4f_c} \psi_l$, for $m = 1,\dots, M'$ and $t^{(l) \star}_{m} = t_{\max}$ for $m = M'+1,\dots, M$ at the $l$th RF chain. 
The PS values are, respectively, given by \eqref{eqPSa} and \eqref{eqPSb}, $x^{(l) \star}_{n,m} = \frac{N-2n+1}{2}\psi_l$, for $m = 1, \dots,M'$ and $x^{(l) \star}_{n,m'} = \vartheta_{\max} - \gamma^{(l)}_{n,m}$, for $m' = M'+1, \dots,M$.   
Then, the array gain $g(\bff_k^{(l)},\psi_{k,l})$ defined in \eqref{eq5} is given by 
\begin{equation}    
\begin{aligned}
&g(\bff_k^{(l)},\psi_{k,l})\\
&= \scalemath{1}{\frac{1}{N_t}\Big|\sum_{m=1}^{M'}\sumN e^{j \pi \xi_k \gamma^{(l)}_{n,m}}e^{j \pi \frac{N-2n+1}{2}f_c} e^{-j\pi \xi_k \frac{(2m-1)N-1}{2}\psi_l}} \\
&+ \scalemath{1}{\sum_{m'=M'+1}^{M} \sumN e^{j \pi \xi_k \gamma^{(l)}_{n,m'}}e^{j \pi (\vartheta_{\max} - \gamma^{(l)}_{n,m'})} e^{-j\pi \xi_k \vartheta_{\max}} \Big|}\\
&= \scalemath{0.85}{\frac{1}{N_t}\Big| M'\frac{\sin(N\Delta_{k,l})}{\sin(\Delta_{k,l})} + e^{-j\pi(\xi_k-1)\vartheta_{\max}} \sum_{m'=M'+1}^{M}\sumN e^{j \pi (\xi_k-1) \gamma^{(l)}_{n,m'}}\Big|}.\label{eqAGlossBound2}
\end{aligned}
\end{equation}   
Applying the triangle inequality to \eqref{eqAGlossBound2} yields   
\begin{equation}
    \label{eq:GainUB}
    \frac{M'}{N_t}\Big|\frac{\sin(N\Delta_{k,l})}{\sin(\Delta_{k,l})}\Big| - \frac{(M-M')N}{N_t} \leq g(\bff_k^{(l)},\psi_{k,l}),
\end{equation}
where $\frac{M'}{N_t}\Big|\frac{\sin(N\Delta_{k,l})}{\sin(\Delta_{k,l})}\Big|$ is the array gain achieved by the first $M'$ TTD precoding $\{t^{(l) \star}_{m}\}_{m=1}^{M'}$ and $\frac{(M-M')N}{N_t}$ is an upper bound of the array gain obtained by the last $(M-M')$ TTDs precoding $\{t^{(l)\star}\}_{m= M'+1}^{M}$. 
From \eqref{eq:GainUB} and noting that $N_t = MN$, the array gain loss $\Big|\frac{\sin(N\Delta_{k,l})}{N\sin(\Delta_{k,l})}\Big| - g(\bff_k^{(l)},\psi_{k,l})$ is bounded by 
\begin{equation}
\begin{aligned}
\label{eq:GainUB2}
&\scalemath{0.9}{0 \leq \Big| \frac{\sin(N\Delta_{k,l})}{N\sin(\Delta_{k,l})} \Big| - g(\bff_k^{(l)},\psi_{k,l})} \\ &~~\scalemath{0.9}{\leq \scalemath{1}{\frac{M-M'}{M}\Big|\frac{\sin(N\Delta_{k,l})}{N\sin(\Delta_{k,l})}\Big| + \frac{(M-M')}{M}},}
\end{aligned}
\end{equation}
which completes the proof.}
\section{Proof of theorem~\ref{theorem2}}
  \label{AppendixD}    
   The second constraint in \eqref{eqMopt1} can be rewritten as  
\begin{equation}
\begin{aligned}
        g_0 & \leq 1+\frac{1}{6}\left(1-\frac{N^2_t}{M^2}\right) \max_{k,l} \left(\frac{\pi}{2} \frac{B}{f_c} \left(\frac{k-1-\frac{K-1}{2}}{K}\right)\psi_l\right)^2 \\ &= 1+\frac{1}{6}\left(1-\frac{N^2_t}{M^2}\right)\left(\frac{\pi}{2} \frac{B}{f_c} \frac{K-1}{2K}\right)^2 \max_l \psi^2_l, \label{eqTaylor2c}
\end{aligned}
\end{equation}
where the upper bound in \eqref{eqTaylor2c} follows from the substitution of $\Delta_{k,l}$ $=$ $\frac{\pi}{2}(\zeta_k-1)\psi_l$ with $\zeta_k$ defined in \eqref{eqSDc} and the last equality in \eqref{eqTaylor2c} is due to the substitutions $k = K$ and $\psi_l^2 = \max_{l} \psi^2_l$. Then, the inequality in \eqref{eqTaylor2c} can be rewritten as $ M \geq \sqrt{\frac{N^2_t}{1+\Omega(g_0,B)}}$, where $\Omega(g_0,B) =\frac{6(1-g_0)}{\left(\frac{\pi}{2} \frac{B}{f_c} \frac{K-1}{2K}\right)^2 \max_l \psi^2_l }$, leading to ${M}^{\star}= \left\lceil \sqrt{\frac{N^2_t}{1+\Omega(g_0,B)}} \right\rceil_{N_t}$ because $M \in \cN_t$. This completes the proof.
\bibliographystyle{IEEEtran} 
\bibliography{biblib} 
%% qua's bio
\begin{IEEEbiography}
[{\includegraphics[width=1in,height=1.25in,clip,keepaspectratio]{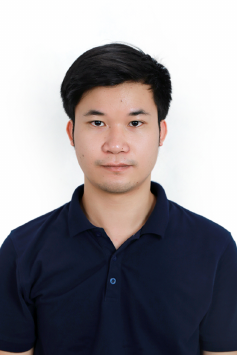}}]{Dang Qua Nguyen} (Student Member, IEEE)
 received a B.Eng. degree in Electronics and Telecommunications from Hanoi University of Science and Technology (HUST), Hanoi, Vietnam in 2019. 
 He is currently pursuing a Ph.D. degree in Electrical Engineering and Computer Science at The University of Kansas (KU), KS, USA, starting from September 2020. 
 His research interests are statistical signal processing, wireless communication, and federated learning.
\end{IEEEbiography}
%% Prof. Kim's bio
\begin{IEEEbiography}
[{\includegraphics[width=1in,height=1.25in,clip,keepaspectratio]{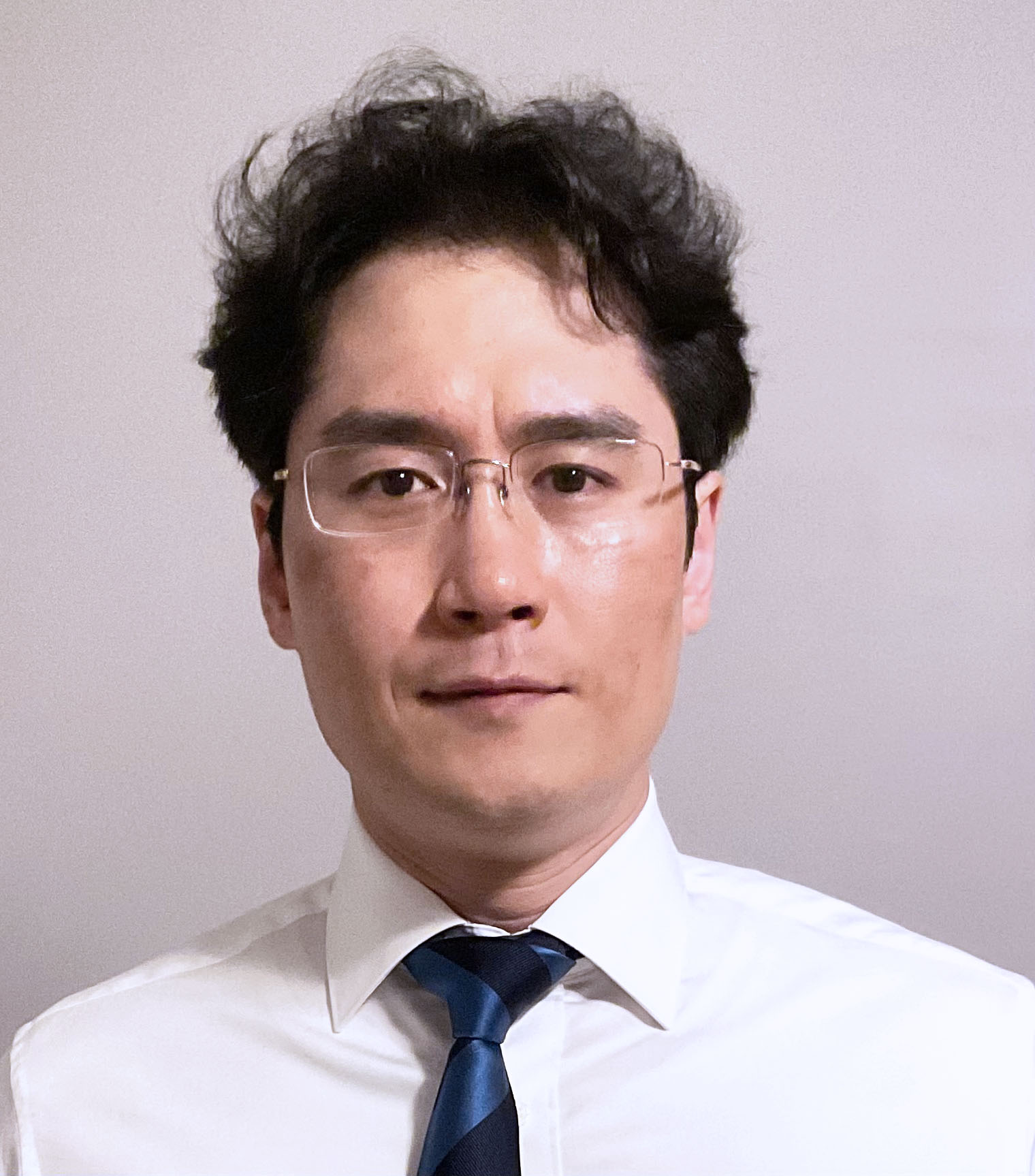}}]{Taejoon Kim}
(Senior Member, IEEE) received the Ph.D. degree in electrical and computer engineering from Purdue University, West Lafayette, IN, USA. He is currently an Associate Professor in the School of Electrical, Computer and Energy Engineering, Arizona State University (ASU). Before joining ASU, he held positions with Nokia Bell Laboratories, KTH Royal Institute of Technology, City University of Hong Kong, and University of Kansas (KU). He is an Associate Editor for the IEEE Transactions on Wireless Communications and served as an Associate Editor of the IEEE Transactions on Communications and Guest Editor for the IEEE Transactions on Industrial Informatics. He holds 29 issued U.S. patents. His research interests are in the design and analysis of wireless communication systems, Next-G wireless systems, multiple-input multiple-output (MIMO) communications, millimeter wave MIMO, statistical signal processing, distributed and reinforcement learning, and network security. He was a recipient of the Harry Talley Excellence in Teaching Award, the Faculty of the Year Award, and the Miller Professional Development Award in Research. Along with the coauthors, he won the IEEE Communications Society Stephen O. Rice Prize in 2016 and the IEEE PIMRC 2012 Best Paper Award. 
\end{IEEEbiography}
\end{document}